\documentclass[11pt]{article}
\usepackage{fullpage}
\usepackage{rpmacros}
\RequirePackage[colorlinks=true]{hyperref}
\hypersetup{
  linkcolor=[rgb]{0.3,0.3,0.6},
  citecolor=[rgb]{0.2, 0.6, 0.2},
  urlcolor=[rgb]{0.6, 0.2, 0.2}
}
\usepackage{mathpazo}
\usepackage{bbm}
\usepackage{todonotes}
\usepackage{lipsum}
\usepackage{setspace}

\usepackage{amsthm}
\usepackage{thmtools,thm-restate}

\numberwithin{equation}{section}
\declaretheoremstyle[bodyfont=\it,qed=\qedsymbol]{noproofstyle}

\declaretheorem[numberlike=equation]{observation}

\declaretheorem[name=Observation,numbered=no]{observation*}

\declaretheorem[numberlike=equation]{theorem}
\declaretheorem[numberlike=equation,style=noproofstyle,name=Theorem]{theoremwp}
\declaretheorem[name=Theorem,numbered=no]{theorem*}

\declaretheorem[numberlike=equation]{lemma}
\declaretheorem[name=Lemma,numbered=no]{lemma*}

\declaretheorem[numberlike=equation]{corollary}
\declaretheorem[name=Corollary,numbered=no]{corollary*}
\declaretheorem[numberlike=equation,style=noproofstyle,name=Corollary]{corollarywp}

\declaretheorem[name=Proposition,numbered=no]{proposition*}

\declaretheorem[numberlike=equation]{claim}
\declaretheorem[name=Claim,numbered=no]{claim*}

\declaretheorem[name=Conjecture,numbered=no]{conjecture*}

\declaretheorem[name=Question,numbered=no]{question*}

\declaretheoremstyle[bodyfont=\it,qed=$\lozenge$]{defstyle} 

\declaretheorem[numberlike=equation,style=defstyle]{definition}
\declaretheorem[unnumbered,name=Definition,style=defstyle]{definition*}

\declaretheorem[unnumbered,name=Example,style=defstyle]{example*}

\declaretheorem[unnumbered,name=Notation=defstyle]{notation*}

\declaretheorem[unnumbered,name=Construction,style=defstyle]{construction*}

\declaretheorem[numberlike=equation,style=defstyle]{remark}
\declaretheorem[unnumbered,name=Remark,style=defstyle]{remark*}

\def\rank{\mathsf{Rank}}
\def\dim{\mathsf{Dim}}
\def\eval{\mathsf{Eval}}
\def\span{\mathsf{Span}}
\def\mult{\mathsf{Mult}}
\def\deg{\mathsf{Deg}}

\def\terms{\mathrm{Terms}}
\def\size{\mathrm{size}}

\def\NW{\mathsf{NW}}

\def\Perm{\mathsf{Perm}}
\def\Det{\mathsf{Det}}

\title{An exponential lower bound for homogeneous depth-$5$ circuits over finite fields%
{\IfFileExists{./sha.tex}{\\\small SHA: \input{sha}}{}}}
\author{
Mrinal Kumar\thanks{Research supported in part by NSF grant CCF-1253886 and  Simons Graduate Fellowship. Part of this work  done while  an intern at MSR, New England.}\\
Rutgers University\\
\texttt{mrinal.kumar@rutgers.edu}
\and
Ramprasad Saptharishi\thanks{The research leading to these results has received funding from the European Community's Seventh Framework Programme (FP7/2007-2013) under grant agreement number 257575.}\\
Tel Aviv University\\
\texttt{ramprasad@cmi.ac.in}
}
\begin{document}
\maketitle
\onehalfspace
\begin{abstract}
In this paper, we show exponential lower bounds for the class of homogeneous depth-$5$ circuits over all small finite fields. More formally, we show that there is an explicit family $\{P_d : d \in \N\}$ of polynomials in $\VNP$, where $P_d$ is of degree $d$ in $n = d^{O(1)}$ variables, such that over all finite fields $\F_q$, any homogeneous depth-$5$ circuit which computes $P_d$ must have size at least $\exp(\Omega_q(\sqrt{d}))$. 

To the best of our knowledge, this is the first super-polynomial lower bound for this class for any field $\F_q \neq \F_2$. 

Our proof builds up on the ideas developed on the way to proving lower bounds for homogeneous depth-$4$ circuits~\cite{gkks13, FLMS13, KLSS, KS14} and for non-homogeneous depth-$3$ circuits over finite fields~\cite{grigoriev98, gr00}. 
Our key insight is to look at the space of shifted partial derivatives of a polynomial as a space of functions from $\F_q^n \rightarrow \F_q$ as opposed to looking at them as a space of formal polynomials and builds over a tighter analysis of the lower bound of Kumar and Saraf~\cite{KS14}. 
\end{abstract}

\section{Introduction}

Arithmetic circuits are the most natural model to study computations of multivariate polynomials. These are directed acyclic graphs, with a unique sink node called the root or output gate,  internal nodes are labelled by addition and multiplication gates\footnote{throughout this paper, we consider circuits as having gates of unbounded fan-in.}, and leaves are labelled by variables or constants from the underlying field. The field of arithmetic circuit complexity aims at understanding the hardness of multivariate polynomials in terms of the size of the smallest arithmetic circuit computing it. One of the most important questions in this field of study is to show that there are families of explicit \emph{low-degree}\footnote{where the degree is bounded by a polynomial function in the number of variables} polynomials that require arithmetic circuits of super-polynomial size (in terms of $n$, the number of variables). It is widely believed that the symbolic $n\times n$ permanent, often denoted by $\Perm_n$, requires circuits of size $\exp(\Omega(n))$ but, as of now,  we do not even have a $\Omega(n^2)$ lower bound for any explicit polynomial. 

\subsection*{Depth Reductions}

In the absence of much progress on the question of lower bounds for general arithmetic circuits, a natural question to ask is if  we can prove good lower bounds for nontrivial restricted classes of circuits. One particular class of circuits which have been widely studied with this aim are the class of bounded depth\footnote{A depth $k$ arithmetic circuit consists of $k$ layers of alternating sum and multiplication gates with the output being computed by a sum gate.} arithmetic circuits. It turns out that this is not just an attempt to study a simpler model, but there is a formal connection between lower bounds for bounded depth circuits and lower bounds for general circuits. A sequence of structural results, often referred to as {\it depth reduction} results, show that strong enough lower bounds for bounded depth circuits implies lower bounds for general arithmetic circuits.

The first depth reduction for arithmetic circuits was by Hyafil~\cite{Hyafil1978} who showed that any polynomial computed by a polynomial sized arithmetic circuit can be equivalently computed by a circuit of depth $O(\log d)$ and \emph{quasi-polynomial} size. This was improved by  Valiant, Skyum, Berkowitz and Rackoff~\cite{vsbr83}, who showed that any $n$-variate degree $d$ polynomial that can be computed by a circuit of size $(nd)^{O(1)}$ can be equivalently computed by a circuit of depth $O(\log d)$ and size $(nd)^{O(1)}$. Thus, proving super-polynomial lower bounds for $O(\log d)$ depth circuits is sufficient to prove super-polynomial lower bounds for general arithmetic circuits. Agrawal and Vinay~\cite{av08} further strengthened this to obtain a depth reduction to depth-$4$ circuits by showing that any $n$-variate degree $d$ polynomial that can be computed by a $2^{o(n)}$ sized circuit can be equivalently computed by \emph{homogeneous}\footnote{which means that all intermediate computations are homogeneous polynomials. Hence the degree of any intermediate computation is bounded by the degree of the output polynomial.}  depth-$4$ circuit of size $2^{o(n)}$. Their result was strengthened by Koiran~\cite{koiran} and Tavenas~\cite{Tav13} to show that any circuit of size $s$ that computes an $n$-variate degree $d$ polynomial can be computed by a homogeneous depth-$4$ circuit of size $s^{O(\sqrt{d})}$, and in fact the resulting depth-$4$ circuits have all multiplication fan-ins bounded by $O(\sqrt{d})$. These results hold over all fields.

Over any field of characteristic zero, Gupta, Kamath, Kayal and Saptharishi~\cite{gkks13b} showed that any $n$-variate degree $d$ polynomial computed by a size $s$ circuit can be equivalently computed by a non-homogeneous depth-$3$ circuit of size $s^{O(\sqrt{d})}$. 
Thus, these results formally show that proving good enough lower bounds on circuits of bounded depth is sufficient for proving lower bounds for general circuits. 

\subsection*{Lower bounds for depth-$3$ and depth-$4$ circuits}

Nisan and Wigderson~\cite{nw1997} proved an $\exp(\Omega(n))$ lower bound for any \emph{homogeneous} depth-$3$ circuits computing the symbolic $n\times n$ determinant $\Det_n$ by studying dimension of the partial derivatives of $\Det_n$ as polynomials. Grigoriev and Karpinski~\cite{grigoriev98} and Grigoriev and Razborov~\cite{gr00} extended this to prove an $\exp(\Omega(n))$ lower bound for non-homogenous depth-$3$ circuit computing $\Det_n$ over any fixed finite field $\F_q$. Chillara and Mukhopadhyay \cite{CM14a} extended this to give a $\exp(\Omega_q(d \log n))$ lower bound for non-homogeneous depth-$3$ circuits computing an $n$-variate degree $d$ polynomial in $\VP$. It is worth noting that there is no generic method known to convert a boolean lower bound for $\mathsf{AC}^0[\bmod q]$ to lower bounds for arithmetic circuits over $\F_q$ (discussed in more detail in \autoref{sec:AC0-connections}). 

The proofs of \cite{grigoriev98,CM14a} also studied the dimension of partial derivatives of polynomial, but unlike the proof in~\cite{nw1997}, they looked at partial derivatives as functions from $\F_q^n \rightarrow \F_q$.
The proofs in \cite{grigoriev98}, \cite{gr00} and \cite{CM14a} strongly rely on the fact that we are working over small finite fields, and completely break down over larger fields or fields of large characteristic. Over fields of characteristic zero and over algebraic closure of finite fields, the question of proving superpolynomial lower bounds for non-homogeneous depth three circuits continues to remain wide open. 

Even though we had exponential lower bounds for homogeneous depth-$3$ circuits, the question of proving superpolynomial lower bounds for homogeneous depth-$4$ circuits remained open for more than a decade.
In 2012, Kayal~\cite{k2} introduced the notion of  \emph{shifted partial derivatives}, which is a generalization of the well-known notion of partial derivatives. Shifted partial derivatives have been very influential in a plethora of lower bounds for depth-$4$ circuits in the past few years. Gupta et. al. \cite{gkks13} used this measure to prove an $\exp(\Omega(\sqrt{n}))$ lower bound for the size of homogeneous depth-$4$ circuits  with multiplication fan-ins bounded by $O(\sqrt{n})$. Subsequently, lower bounds of $\exp(\Omega(\sqrt{d}\log n))$ were proved for other $n$-variate degree $d$ polynomials computed by almost the  same circuit class \cite{KSS13,FLMS13,KS14a}. (It is worth noting that getting a lower bound of $\exp(\omega(\sqrt{d}\log n))$ would have implied a super-polynomial lower bound for general circuits!) Using a more delicate variant called \emph{projected shifted partials}, Kayal et. al. \cite{KLSS} and Kumar and Saraf \cite{KS14} proved lower bounds of $\exp(\Omega(\sqrt{d}\log n))$ for homogeneous depth-$4$ circuits (without any fan-in restrictions) via two very different analyses. The former was an analytic approach and works only over characteristic zero fields, whereas the latter was purely combinatorial and works over any field. These techniques have also been applied to yield lower bounds for  non-homogeneous depth-$3$ circuits with bounded bottom fan-in~\cite{KayalSaha14} and homogeneous depth-$5$ circuits with bounded bottom fan-in~\cite{BC15}. A continuous updated survey \cite{github} contains expositions of many of the lower bounds and depth reduction results listed above. 

The results in \cite{KS14} in fact show that the reduction from general arithmetic circuits to depth-$4$ circuits with support $O(\sqrt{d})$ cannot be improved, as they give an example of a polynomial in $\VP$ for which any  depth-$4$ circuits of support $O(\sqrt{d})$ must be of size $n^{\Omega(\sqrt{d})}$. Further, with the current upper-bounds for the projected shifted partials on such depth-$4$ circuits, the best we can hope to prove using this measure is an $n^{\Omega(\sqrt{d})}$ lower bound. Hence, it might be insufficient for general arithmetic circuits lower bounds but it could well be the case that we might be able to prove stronger lower bounds for constant depth arithmetic circuits, or arithmetic formulas by variants of this family of measures. 

Hence, as a start, the problem of proving lower bounds for homogeneous depth five circuits, seems like the next natural question to explore. This already seems to introduce new challenges as the proofs of lower bounds for homogeneous depth-$4$ circuits seem to break down for homogeneous depth-$5$ circuits. In this paper, we pursue this line of enquiry, and prove exponential lower bounds for homogeneous depth-$5$ circuits over small finite fields. Before stating our results, we first discuss prior results on this question, and the challenges involved in extending the proofs of lower bounds for homogeneous depth four circuits, in the next section.

\subsection*{Lower bounds for depth-$5$ circuits}
Prior to this work, the only known lower bounds for depth-$5$ circuits that we are aware of are the results of Raz~\cite{raz10}, which show superlinear lower bounds for bounded depth circuits over large enough fields, the results of Kalorkoti~\cite{k85} which show quadratic lower bounds for arithmetic formulas and the results of Bera and Chakrabarti~\cite{BC15} and Kayal and Saha~\cite{KayalSaha14} which show  exponential lower bounds for homogeneous depth-$5$ circuits if the bottom fan-in is bounded.  

Given that we have lower bounds for homogeneous depth-$4$ circuits, it seems natural to try and apply these techniques to prove lower bounds for homogeneous depth-$5$ circuits. Unfortunately, the obvious attempts to generalize the proofs in~\cite{KLSS, KS14} seem to fail for homogeneous depth-$5$ circuits. We now elaborate on this.  

\paragraph{On extending the depth-$4$ lower bound  proofs to depth-$5$ circuits : }
To understand these issues, we first need a birds-eye view of the major steps in the proofs of lower bounds for depth-$4$ circuits~\cite{KLSS, KS14}. These proofs have two major components. 
\begin{itemize}
\item {\bf Reduction to depth-$4$ circuits with bounded bottom support : } In the first step, the circuit $C$ and the polynomial are hit with a random restriction, in which each  variable is kept alive independently with some small probability $p$. The observation is that a bottom level product gate in $C$ of support (the number of distinct variable inputs) at least $s$ survives with probability at most $p^s$. Therefore, the probability that some bottom product of support at least $s$ in $C$ survives is at most $\text{Size}(C)\cdot p^s$. Now, if the size of $C$ is small (say  $\epsilon\cdot 1/p^s$), then this probability is quite small, so with a high probability $C$ reduces to a homogeneous depth-$4$ circuit with bounded bottom support.   
\item {\bf Lower bounds for depth-$4$ circuits with bounded bottom support :} The goal in the second step is to show that the polynomial obtained after random restrictions still remains hard for homogeneous depth-$4$ circuits with bottom support at most $s$. 
\end{itemize}

The key point in step $1$ is that if $\text{Size}(C)$ is not too large, then we can assume that with a high probability over the random restrictions, all the high support product gates are set to $0$. This is where things are not quite the same for depth-$5$ circuits. When we express a homogeneous depth-$5$ circuit as a homogeneous depth-$4$ circuit by expanding the product of linear forms at level four, we might increase the number of monomials a lot (potentially to all possible monomials). Now, the random restriction step no longer works and we do not have a reduction to homogeneous depth-$4$ circuits with bounded bottom support. If the bottom fan-in of $C$ is bounded, then this strategy does indeed generalize. Bera and Chakrabarti~\cite{BC15} and Kayal and Saha~\cite{KayalSaha14} show exponential lower bounds for such cases. 

It is not clear to us how fundamental this obstruction is, but our key insight  is a  strategy for proving lower bounds for  homogeneous depth-$4$ circuits that avoids the random restriction step. Morally speaking, we \emph{do} proceed by a `reduction' from a depth-$5$ circuit to a depth-$4$ circuit, but the meaning of a `reduction' here is more subtle and largely remains implicit.

\subsection*{Our Contribution}

We give  an exponential lower bound for homogeneous depth-$5$ circuits over any fixed finite field $\F_q$.
To the best of our understanding, this is the first such lower bound for depth-$5$ circuits over any field apart from $\F_2$\footnote{For $\F_2$, exponential lower bounds easily follow from the lower bounds of Razborov~\cite{razborov87} }. Stated precisely, we prove the following theorem. 

\begin{theorem}\label{thm:mainthm-all fields}
There is an explicit family of polynomials $\setdef{P_d}{d\in \N}$, with $\deg(P_d) = d$, in the class $\VNP$ such that for any finite field $\F_q$, any homogeneous depth-$5$ circuit computing $P_d$ must have size $\exp(\Omega_q(\sqrt{d}))$. 
\end{theorem}

The polynomial $P_d$ is from the Nisan-Wigderson family of polynomials (introduced by \cite{KSS13}, \autoref{def:NW final}) with carefully chosen parameters. 

Our proof also extends to non-homogeneous depth-$5$ circuits where the layer of multiplication gates closer to the output have fan-in bounded by $O(\sqrt{d})$ (with no restriction on the fan-in of the other multiplication layer). 

\begin{theorem}\label{thm:mainthm-nonhom}
There is an explicit family of polynomials $\setdef{P_d}{d\in \N}$, with $\deg(P_d) = d$, in the class $\VNP$ such that for any finite field $\F_q$, any $\Sigma\Pi^{[O(\sqrt{d})]}\Sigma\Pi\Sigma$ circuit computing $P_d$ must have size $\exp(\Omega_q(\sqrt{d}))$. 
\end{theorem}

It is worth mentioning that for characteristic zero fields, it suffices to prove an $\exp(\omega(d^{1/3}\log d))$ lower bound for an explicit polynomial computed by such $\Sigma\Pi^{[O(\sqrt{d})]}\Sigma\Pi\Sigma$ circuits to separate $\VP$ from $\VNP$ (by combining the depth reductions of \cite{av08,koiran,Tav13} and \cite{gkks13b}). We elaborate on this in \autoref{sec:lb-char-zero}. Such a phenomenon also happens for non-homogeneous depth three circuits, where over finite fields, we know quite strong lower bounds while much weaker ones would imply $\VNP \neq \VP$ over fields of characteristic zero.\\

The key technical ingredient of our proof is to look at the space of shifted partial derivatives and the projected shifted partial derivatives of a polynomial. We study them as a space of functions from $\F_q^n \rightarrow \F_q$ as opposed to as a space of formal polynomials, as has been the case for the results obtained so far. This perspective allows us the freedom to confine our attention to the evaluations of the shifted partial derivatives of a polynomial on certain well chosen subsets of $\F_q^n$, and this turns out to be critical to our cause. This leads to a new family of complexity measures which could have applications to other lower bound questions as well. Our proof also involves a tighter analysis of the lower bound of Kumar and Saraf~\cite{KS14} (for homogeneous depth-$4$ circuits) which may be interesting in its own right. \\

\noindent
We now give an overview of our proof. 

\section{An overview of the proof}~\label{sec:proof sketch} 

\noindent
The proof would consist of the following main steps:

\begin{enumerate}
\item Define a function $\Gamma: \F_q[\vecx] \rightarrow \N$.  Intuitively, we think of $\Gamma(P)$ to be a measure of the {\it complexity} of $P$.
\item For all homogeneous depth-$5$ circuits $C$ of size at most $\exp(\delta \sqrt{d})$, prove an upper bound on $\Gamma(C)$. 
\item For the target hard polynomial $P$, show that $\Gamma(P)$ is much larger than the upper bound proved in step 2. 
\end{enumerate}

\paragraph{The complexity measure :} At a high level, the proof of lower bounds in~\cite{nw1997, gkks13, KSS13, FLMS13, KS14a, KLSS, KS14} associate a linear space polynomials to every polynomial in $\F_q[\vecx]$ and use the dimension of this space over $\F_q$ as a measure of complexity of the polynomial. The mapping from polynomials to linear space of polynomials undergoes subtle changes as we go from the proof of lower bounds for homogeneous depth-$3$ circuits~\cite{nw1997} to lower bounds for homogeneous depth-$4$ circuits~\cite{KLSS, KS14}. 

In this paper, we follow this outline and associate to every polynomial, the space of its shifted partial derivatives as defined in~\cite{gkks13}. However, instead of working with this space of polynomials as it is, we study their evaluation vectors over a subset of $\F_q^n$ (similar to \cite{grigoriev98,gr00}, where they worked with partial derivatives of a polynomial). The key gain that we have from this change in outlook is that as evaluation vectors, we can choose to confine our attention to evaluations on certain properly chosen subsets of $\F_q^n$. For formal polynomials, it is not clear what should be the correct analog of this approximation. The necessity and the utility of this will be more clear as we go along. 
\paragraph{High rank products of linear forms :}
Consider a polynomial $Q$ which is a product of $\tau$ linearly independent linear forms $L_1, L_2, \ldots, L_{\tau}$. 
$$Q = \prod_{i = 1}^{\tau} L_i $$
It is not hard to see that 
$$\Pr_{\veca \in \F_q^n}[Q(\veca) \neq 0] \leq \left(1-\frac{1}{q} \right)^{\tau}$$
In other words, products of linear forms of rank $\tau$ vanish on all but a $o(1)$ fraction of the entire space if $\tau = \omega(1)$. If the size of a depth-$5$ circuit is not too large as a function of $\tau$ (say, at most $\exp(\delta \tau)$ for a small enough $\delta > 0$), then by a union bound, all the products of rank at least $\tau$ at the fourth level vanish everywhere apart from a $o(1)$ fraction of the points in $\F_q^n$. 

In summary, we just argued that a depth-$5$ circuit $C$ over $\F_q$ of size at most $\exp(\delta \tau)$ can be approximated by a sub-circuit $C'$ of $C$ which is obtained from $C$ by dropping all products of linear forms of rank at least $\tau$ from the bottom level. 

\paragraph{Low rank products of linear forms :}
A second simple observation (\autoref{lem:low-rank-to-low-degree}) shows that for every  product of linear forms of rank at most $\tau$, there is a polynomial of degree at most $(q-1)\tau$, such that they agree at all points in $\F_q^n$. Thus, the circuit $C'$  is equal, as a function from $\F_q^n \rightarrow \F_q$ to a depth-$4$ circuit $C''$ of bottom fan-in at most $(q-1)\tau$.  
Moreover, the formal degree and the top fan-in of  $C''$ are upper bounded by the formal degree and top fan-in of $C$, respectively. 

\paragraph{Putting things together :}
This implies that for every homogeneous depth-$5$ circuit $C$ computing a polynomial of degree $d$ of size at most $\exp(\delta \tau)$ for some $\tau$, there exists a depth-$4$ circuit $C''$ of formal degree at most $d$ and top fan-in at most $\exp(\delta \tau)$ such that 

$$\Pr_{\veca \in \F_q^n}[C(\veca) \neq C''(\veca)] \leq o(1).$$

Therefore, a polynomial $P$ which can be computed by $C$ can be {\it approximated} by $C''$ in the pointwise sense. Since we know lower bounds on the top fan-in of homogeneous (and low formal degree) depth-$4$ circuits with bounded bottom fan-in~\cite{gkks13, KSS13}, it seems that we only have a small way to go. Unfortunately, we do not quite know how to make this idea work. The key technical obstacle here is that it seems to be hard to say much about the partial derivatives of $C$ by looking at partial derivatives of $C''$. As a pathological case, the polynomial $\prod_{i \in [n]} x_i$ has a partial derivative span of size $2^n$ but  is well approximated by the $0$ polynomial over $\F_2$. 

If we had started with a depth-$3$ circuit instead of a depth-$5$ circuit, then such a strategy is indeed known to work~\cite{gr00}. Observe that in this case it is enough to show that that there is an explicit polynomial which cannot be approximated well by a low degree polynomial over $\F_q$. In~\cite{gr00}, the authors show this by an adaptation of a similar result of Smolenksy~\cite{smolensky87} over $\F_2$.

\paragraph{A strengthening of the strategy :}
The key additional observation that helps us make things work is the fact that not only do high rank product gates at level four of $C$ vanish almost everywhere on $\F_q^n$, but they vanish with a high multiplicity. As we show in \autoref{cor:corollarywp}, if the size of $C$ is not \emph{too large}, all the product gates of rank at least $\tau$ vanish with a multiplicity $\Omega(\tau)$ at $1-o(1)$ fraction of points on  $\F_q^n$\footnote{In the rest of this discussion, we will think  of $\tau$ as $\Theta(\sqrt{d})$}. 

Therefore, not only can $C$ agree with $C'$ almost everywhere, all the partial derivatives of $C$ of order at most $k = \Omega(\tau)$ agree with all the partial derivatives of $C'$ almost everywhere. This already hints at the fact that if we are to take advantage of this then we should be looking at the evaluation of these derivatives, since our only guarantee is about their evaluations. 

In \cite{grigoriev98}, exponential lower bounds were proved for non-homogeneous depth-$3$ circuits using a very similar strategy. However, adapting the method for shifted partials and projected shifted partials seems to be a challenge. 

In \autoref{sec:circuit upper bound}, we show that the dimension of the span of evaluation vectors of shifted partial derivatives of $C$, when restricted to a properly chosen subset $S$ of $\F_q^n$, is small if the size of the circuit $C$ we started with was small. 

As a final step of the proof, we show that with respect to this complexity measure, our target hard polynomial (from the so-called \emph{Nisan-Wigderson} family, defined below) has a large complexity. 

\begin{definition}[Nisan-Wigderson polynomial families]~\label{def:NW final}
Let $d,m,e$ be arbitrary parameters with $m$ being a power of a prime, and $d,e\leq m$. 
Since $m$ is a power of a prime, let us identify the set $[m]$ with the field $\F_m$ of $m$ elements. 
Note that since $d \leq m$, we have that $[d] \subseteq \F_m$. 
The Nisan-Wigderson polynomial with parameters $d,m,e$, denoted by $\NW_{d,m,e}$ is defined as
\[
\NW_{d,m,e}(\vecx) \spaced{=} \sum_{\substack{p(t) \in \F_m[t]\\ \deg(p) < e}} x_{1,p(1)}\dots x_{d,p(d)}
\]
That is, for every univariate polynomial $p(t) \in \F_m[t]$ of degree less thatn $e$, we add one monomial that encodes the `graph' of $p$ on the points $[d]$. 
This is a homogeneous, multilinear polynomial of degree $d$ over $dm$ variables with exactly $m^e$ monomials. 
\end{definition}

This step of the proof builds on a tighter analysis of the lower bound on the dimension of the span of {\it projected shifted partial derivatives} of the Nisan-Wigderson polynomials  in~\cite{KS14}. We show that if the set $S$ is carefully chosen, then we do not incur much loss in the dimension by going from looking at shifted partial derivatives as formal polynomials to looking at them as functions over a small subset of the  finite field.    
We provide the details in \autoref{sec:poly lower bound}. 

One important technicality is  the dependency between various parameters involved. For our proof, the choice of $k$ would be $O_q(\tau)$ and would depend on $q$. The lower bound of \cite{KS14} would then choose specific parameters for the $\NW_{d,m,e}$. This would mean that for every $q$, we get a \emph{different} polynomial for which we show a lower bound. We remedy the order of quantifiers and start by fixing specific parameters for $\NW_{d,m,e}$. Then,  depending on $q$, we choose a restriction of $\NW_{d,m,e}$ that would be identical to $\NW_{d',m,e}$ by setting some variables to $0/1$. We then apply the \cite{KS14} argument for this restriction to obtain our lower bound for $\NW_{d',m,e}$ which also yields a lower bound for $\NW_{d,m,e}$. The details are in \autoref{sec:order-of-quantifiers}.

\section{Notation}

\begin{itemize}
\item Throughout the paper, we shall use bold-face letters such as $\vecx$ to denote a set $\set{x_1,\dots, x_n}$. 
Most of the times, the size of this set would be clear from context. 
We shall also abuse this notation to use $\vecx^\vece$ to refer to the monomial $x_1^{e_1}\cdots x_n^{e_n}$. 

\item For an integer $m > 0$, we shall use $[m]$ to denote the set $\set{1,\dots, m}$. 

\item We shall use the short-hand $\partial_{\vecx^{\vece}}(P)$ to denote
\[
\frac{\partial^{e_1}}{\partial x_1^{e_1}}\inparen{ \frac{\partial^{e_2}}{\partial x_2^{e_2}}\inparen{\cdots \inparen{ P }\cdots}}.
\]

\item For a set of polynomials $\mathcal{P}$ shall use $\partial^{=k}\mathcal{P}$ to denote the set of all $k$-th order partial derivatives of polynomials in $\mathcal{P}$, and $\partial^{\leq k}\mathcal{P}$ similarly. 

Also, $\vecx^{=\ell} \mathcal{P}$ shall refer to the set of polynomials of the form $\vecx^{\vece} \cdot P$ where $\deg(\vecx^{\vece}) = \ell$ and $P \in \mathcal{P}$. Similarly $\vecx^{\leq \ell} \mathcal{P}$.  

\item For a polynomial $P \in \F_q[\vecx]$ and for a set $S \subseteq\F_q^n$, we shall denote by $\eval_S(P)$ the vector of the evaluation of $P$ on points in $S$ (in some natural predefined order like say the lexicographic order). 
For a set of vectors $V$, their span over $\F_q$ will be denoted by $\span(V)$ and their dimension by $\dim(V)$. 

\item We shall use $\mathcal{H}$ to denote the set $\set{0,1}^n \subset \F_q^n$.
\end{itemize}

\subsection*{The complexity measure}

We now define the complexity measure that we shall be using to prove the lower bound. 
The measure will depend on a carefully chosen set $S \subset \F_q^n$. 

\begin{definition}[The complexity measure]
Let $k, \ell$ be some parameters and let $S \subset \F_q^n$. For any polynomial $P$, define $\Gamma_{k,\ell,S}(P)$ as
\[
\Gamma_{k,\ell,S}(P) \spaced{:=} \dim\inbrace{\eval_S\inparen{\vecx^{=\ell} \partial^{=k}(P)}}.\qedhere
\]
\end{definition}

\section{Complexity measure on a depth-$5$ circuit}\label{sec:circuit upper bound}

A depth-$5$ circuit computes a polynomial of the form 
\begin{equation}\label{eqn:d5-circuit}
C \spaced{=} \sum_a \prod_b \sum_c \prod_d L_{abcd}
\end{equation}
where each $L_{abcd}$ are linear polynomials. 

\begin{definition}[Terms of a circuit, and rank]
For a depth-$5$ circuit such as \eqref{eqn:d5-circuit}, we shall denote by $\terms(C)$ the set
\[
\terms(C) \spaced{:=} \set{\prod_{d} L_{abcd}}_{a,b,c}
\]
which are all products of linear polynomials computed by the bottommost product gates. \\

For any term $T = \prod_d L_d$, define $\rank(T)$ to be $\dim\set{L_d}_d$, which the maximum number of independent linear polynomials among the factors of $T$. For a depth-$5$ circuit $C$, we shall use $\rank(C)$ to denote $\max_{T \in \terms(C)} \rank(T)$. 

For a parameter $\tau$, we shall use $\terms_{> \tau}(C)$ to refer to terms $T \in \terms(C)$ with $\rank(T) > \tau$. 
\end{definition}

\subsection*{Low rank gates are low-degree polynomials}

The following Lemma, present implicitly in \cite{grigoriev98,gr00}, is a very useful transformation of gates of low-rank (and possibly large degree) when working over a finite field.

\begin{lemma}[\cite{grigoriev98,gr00}]\label{lem:low-rank-to-low-degree}
Let $Q$ be a product of linear polynomials of rank at most $\tau$. Then, there is a polynomial $\tilde{Q}$ of degree at most $(q-1) \cdot \tau$ such that $\tilde{Q}(\veca) = Q(\veca)$ for all $\veca \in \F_q^n$. 
\end{lemma}
\begin{proof}
Without loss of generality, we shall assume that the rank is equal to $\tau$, as the degree upper bound will only be better for a smaller rank and let $L_1,\dots, L_\tau$ be linearly independent. 
Let 
$$Q = \prod_{i = [\tau]} L_i \cdot \prod_{j \notin [\tau]} L_j $$
Here, each linear form in the second product term is in the linear span of the linear forms $\{L_i : i \in [\tau]\}$, and so can be expressed as their linear combination. Therefore, $Q$ can be expressed as a polynomial in $\{L_i : i \in [\tau]\}$. Let $Q = f(L_1, L_2, \ldots, L_{\tau})$. Since we are working over $\F_q$, it follows that for every choice of $L_i$ and for every $\veca \in \F_q^n$, we have $L_i^q(\veca) = L_i(\veca)$. So, for every $\veca \in \F_q^n$, 
$$ 
f(L_1, L_2, \ldots, L_{\tau})(\veca) = [f(L_1, L_2, \ldots, L_{\tau}) \mod \inangle{\setdef{L_i^q-L_i}{i=1,\dots, \tau}} ](\veca)
$$
The lemma immediately follows by setting $\tilde{Q} := f(L_1, L_2, \ldots, L_{\tau}) \mod \inangle{\setdef{L_i^q-L_i}{i=1,\dots, \tau}}$

\end{proof}

\subsection*{High rank gates are almost always zero}

Let us assume that $\size(C) \leq 2^{\sqrt{d}/100}$. We shall fix a threshold $\tau$ and call all terms $T$ with $\rank(T) > \tau$ as ``high rank terms'' and the rest as ``low rank terms''. Under a random evaluation in $\F_q^n$, every non-zero linear polynomial takes value zero with probability $1/q$. Thus, if we have a term that is a product of \emph{many} independent linear polynomials, then with very high probability \emph{many} of them will be set to zero, i.e. the term will vanish with high multiplicity at most points. This is formalized by the following definition and the lemma after it. 

\begin{definition}[Multiplicity at a point] For any polynomial $P$ and a point $\veca \in \F_q^n$, we shall say that $\veca$ vanishes \emph{with multiplicity $t$} on $P$ if $Q(\veca) = 0$ for all $Q \in \partial^{\leq t-1}(P)$. In other words, $\veca$ is a root of $P$ and all its derivatives up to order $t-1$.  

We shall denote by $\mult(P,\veca)$ the maximum $t$ such that $\veca$ vanishes on $\partial^{\leq t-1}(P)$. 
\end{definition}

It is easy to see that if $P$ is a product of linear polynomials, then $\veca$ vanishes with multiplicity $t$ on $P$ if $\veca$ vanishes on at least $t$ factors of $P$. 

\begin{observation}~\label{obs:high multiplicity zeros}
Let $T = \prod_{i=1}^d L_i$ be a term of rank at least $r$. 
Then, for every $\delta > 0$,  
\[
\Pr_{\veca \in \F_q^n}\insquare{\mult(T, \veca) \leq(1-\delta)\frac{r}{q} } \leq \exp{\left(-\frac{\delta^2 r}{2q}\right)}.
\]
\end{observation}   
\begin{proof}
Without loss of generality, let $L_1,\dots, L_r$ be linearly independent. Then, the evaluations of $L_1,\dots, L_r$ at a point $\veca \in \F_q^n$ are also linearly independent and $\Pr_{\veca} [L_i(\veca) = 0] = (1/q)$ for $i = 1,\ldots, r$. 

For $i = 1,\dots, r$, let $Y_i$ be the indicator random variable that is one if $L_i(\veca) = 0$ and zero otherwise. Let $Y = \sum_{i \in [r]} Y_i$.
Clearly, by linearity of expectations $$\E[Y] = \sum_{i \in [r]} \E[Y_i] = \frac{r}{q}.$$
Since the events $Y_i$ are linearly independent, by the Chernoff Bound, we know that for every $\delta > 0$ \[\Pr\left[Y \leq (1-\delta)\frac{r}{q} \right] \spaced{\leq} \exp{\left(-\frac{\delta^2 r}{2q}\right)}. \qedhere
\]
\end{proof}

\noindent
The following corollary is a simple union bound on all high-rank gates in a small circuit. 

\begin{corollarywp}~\label{cor:corollarywp}
Let $C$ be a depth-$5$ circuit over $\F_q$ such that $\size(C) \leq 2^{\sqrt{d}/100}$. Let $\tau = O(\sqrt{d})$ so that 
\[
\exp\inparen{\frac{\tau}{8 \cdot q}} \spaced{<} 2^{\sqrt{d}/50}. 
\]
Then,
\[
\Pr_{\veca \in \F_q^n} \insquare{\exists T \in \terms_{> \tau}(C)\;:\;\mult(T, \veca) \;\leq\; \frac{\tau}{2q}} \spaced{\leq} 2^{-(\sqrt{d}/50)}\qedhere
\]
\end{corollarywp}

\noindent
We shall set our parameter $\tau$ as in the above corollary and our parameter $k = \tau / 2q^3$. 

\subsection{Upper bound on complexity measure}\label{subsec:upper-bound}

For a circuit $C$ of size at most $2^{\sqrt{d}/100}$, let $\mathcal{E}$ refer to the ``bad set'' of points $\veca$ such that there is some $T\in \terms_{> \tau}(C)$ for which $\mult(T,\veca) \leq k = \tau/2q^3$. By the above corollary, we know that 
\[
\abs{\mathcal{E}} = \delta \cdot q^n\quad\text{for some }\delta = \exp(-O(\sqrt{d})). 
\]
Let $S$ be any subset of $\F_q^n \setminus \mathcal{E}$ that is contained in a ``translate of a hypercube'', that is there exists some $\vecc \in \F_q^n$ such that 
\[
S \subset (\vecc + \mathcal{H}) \setminus \mathcal{E}. 
\]
The following lemma allows us to ``multilinearize'' any polynomial as long as we are only interested in evaluations on a translate of a hypercube. 

\begin{lemma}[Multilinearization]\label{lem:multilinearization}
Fix a translate of a hypercube $\vecc + \mathcal{H}$. Then for every polynomial $Q \in \F_q[\vecx]$, there is a unique multilinear polynomial $Q'$ such that $\deg(Q') \leq \deg(Q)$ and $Q'(\veca) = Q(\veca)$ for every $\veca \in \vecc + \mathcal{H}$. 
\end{lemma}
\begin{proof}
If $\veca \in \vecc + \mathcal{H}$, then for each $i \in [n]$ we have $a_i$ to be either $c_i$ or $c_i + 1$. Thus, it suffices to replace each $x_i^2$ by a linear polynomial in $x_i$ that maps $c_i$ to $c_i^2$ and $c_i+1$ to $(c_i+1)^2$. This is achieved by $x_i^2 \mapsto c_i^2 + (x_i - c_i)(2c_i + 1)$. By repeated applications of this reduction, we obtain a multilinear polynomial $Q'$ of degree at most $\deg(Q)$ that agrees on all points on $\vecc + \mathcal{H}$. 

Another way to state this is by looking at $Q \bmod \mathcal{I}_\vecc$  where $\mathcal{I}_\vecc$ is the ideal defined by
\[
\mathcal{I}_\vecc \spaced{:=} \inangle{\setdef{x_i^2 - (c_i^2 + (x_i - c_i)(2c_i + 1))}{i=1,\dots, n}}. 
\]
It is easy to check that $\mathcal{I}_\vecc$ vanishes on $\vecc + \mathcal{H}$, and any $Q$ can be reduced to a multilinear polynomial modulo $\mathcal{I}_\vecc$. 

The uniqueness of $Q'$ follows from the fact that no non-zero multilinear polynomial can vanish on all of $\vecc + \mathcal{H}$. 
\end{proof}

The main lemma of this theorem would be the following bound on the complexity measure on a depth-$5$ circuit. 

\begin{lemma}[Upper bound on circuit]\label{lem:upper-bound-circuit}
Let $C$ be a depth-$5$ circuit, of formal degree at most $2d$ and $\size(C) \leq 2^{\sqrt{d}/100}$, that computes an $n$-variate degree $d$ polynomial. Let $\tau$ and $k$ be chosen as above, and $\ell$ be a parameter satisfying $\ell + k\tau q < n/2$. If $S$ is any subset of $\F_q^n \setminus \mathcal{E}$ that is contained in a translate of a hypercube, then
\[
\Gamma_{k,\ell,S}(C) \spaced{\leq} 2^{\sqrt{d}/100} \cdot \binom{\frac{4d}{\tau} + 1}{k} \cdot \binom{n}{\ell + k\tau q}\cdot \poly(n).
\]
\end{lemma}
\begin{proof}
Suppose $C = R_1 + \cdots + R_s$, where $s \leq 2^{\sqrt{d}/100}$ and each $R_i$ is a product of depth-$3$ circuits with $\deg(R_i) \leq 2d$. Since $\Gamma_{k,\ell,S}$ is clearly sub-additive (i.e. $\Gamma_{k,\ell,S}(f+g) \leq \Gamma_{k,\ell,S}(f) + \Gamma_{k,\ell,S}(g)$ for any $f,g$), it suffices to show that for each $R_i$ we have
\[
\Gamma_{k,\ell,S}(R_i) \spaced{\leq} \binom{\frac{4d}{\tau} + 1}{k} \cdot \binom{n}{\ell + k\tau q}\cdot \poly(n). 
\]
For each such $R_i$, define the $R_i^{\leq \tau}$ as the polynomial obtained by ``deleting'' all terms $T \in \terms_{>\tau}(R_i)$. That is,
\[
\text{if} \quad R_i \spaced{=} \prod_a \sum_b T_{ab} \quad\text{then}\quad R_i^{\leq \tau} \spaced{=} \prod_a\; \sum_{b: Rank(T_{ab}) \leq \tau} T_{ab}.
\]
The lemma follows from the following two claims whose proofs shall be deferred to the end of this section. 
\vspace{-20pt}
\begin{quote}
\begin{claim}\label{clm:internal 1}
For every $i \in [r]$ $$\Gamma_{k,\ell,S}(R_i) \spaced{=} \Gamma_{k,\ell,S}(R_i^{\leq \tau})$$ 
\end{claim}
\begin{claim}\label{clm:internal 2}
For every $i \in [r]$ 
 $$\Gamma_{k,\ell,S}(R_i^{< \tau}) \spaced{\leq} \binom{\frac{4d}{\tau}+1}{k} \cdot \binom{n}{\ell + k\tau q} \cdot \poly(n)$$
\end{claim}
\end{quote}
The lemma readily follows from \autoref{clm:internal 1} and \autoref{clm:internal 2}. 
\end{proof}

\begin{proof}[Proof of \autoref{clm:internal 1}] 
For brevity, we shall drop some indices and work with $R = Q_1 \cdots Q_m$. 
Let $T \in \terms_{>\tau}(C)$. We shall show if $R' = (Q_1 - T) Q_2 \cdots Q_m$, then for any $k$-th order partial derivative $\partial_{\vecx^\alpha}$, 
\[
\eval_{S}(\partial_{\vecx^\alpha}(R)) \quad=\eval_S(\partial_{\vecx^\alpha}(R')). 
\]
Indeed, consider the difference $R - R' = T \cdot Q_2 \cdots Q_m$. 
By the chain rule, every term in $\partial_{\vecx^\alpha}(R - R')$ is divisible by some $k'$-th order partial derivative of $T$ with $k' \leq k$. 
By the choice of $S$, we know that every $\veca \in S$ satisfies $\mult(T,\veca) > k$, and hence $\veca$ vanishes on $\partial^{\leq k}(T)$ for any $T \in \terms_{>\tau}(C)$. Thus, it follows that $\eval_S(\partial_{\vecx^\alpha}(R - R'))$ is just the zero vector. 

Repeating this argument, we can prune away all terms in $\terms_{>\tau}(C)$ to get that $\eval_S(\partial_\alpha(R)) = \eval_S(\partial_{\vecx^\alpha}(R^{\leq\tau}))$ where $\deg(\vecx^{\alpha}) = k$. Thus, $\Gamma_{k,\ell,S}(R) = \Gamma_{k,\ell,S}(R^{< \tau})$.
\end{proof}

\begin{proof}[Proof of \autoref{clm:internal 2}]
Let $R^{\leq \tau} = Q_1 \cdots Q_d$, with each $Q_i$ being a $\SPS$ circuit. Some of these $Q_i$s could have degree more than $\tau$ although their rank is bounded by $\tau$. 
Without loss of generality, let $Q_1,\cdots, Q_m$ be all the $Q_i$s with $\deg(Q_i) > \tau$, and $Q_{m+1},\dots,Q_d$ have $\deg(Q_i) \leq \tau$. 

We shall modify the ``low-degree'' $Q_i$s by multiplying together any two of them of degree less than $\tau/2$. 
This ensures that at most one of the $Q_i$s may have degree less than $\tau/2$ and all the $Q_i$s have degree at most $\tau$. The sizes of some of the low-degree $Q_i$s do increase in the process but this would not be critical as the degree of any such term is still bounded by $\tau$. The main point is that now we have an expression of the form 
\[
R^{\leq \tau} = Q_1 \cdots Q_m \cdot Q_1'\cdots Q_r'
\]
where each $Q_i$ is a $\SPS$ circuit of rank at most $\tau-1$ and $\deg(Q_i) \geq \tau$, and all but one of the $Q_i'$ satisfies $\tau \geq \deg(Q_i') \geq \tau/2$. 
As $\deg(R^{\leq \tau}) \leq 2d$, it follows that $m + r \leq \frac{4d}{\tau} + 1$.

As a  notational convenience, for any set $H$ let  $Q_H := \prod_{i\in H}Q_i$. 
Let us look at any partial derivative $\partial_{\vecx^\alpha}$ of order $k$ applied on $R$. By the chain-rule, any such partial derivative can be written seen as a natural linear combination of terms. 
\begin{eqnarray*}
\partial_{\vecx^\alpha}(R) & \in & \span\setdef{\partial_{\vecx^\beta}(Q_A) \cdot \partial_{\vecx^\gamma}(Q'_B) \cdot Q_{\overline{A}} \cdot Q'_{\overline{B}}}{\begin{array}{c}\vecx^{\alpha} = \vecx^\beta \cdot \vecx^\gamma\;,\; A\subseteq [m]\;,\; \\ B \subseteq [r]\;,\; |A| + |B| = k\end{array}}\\
& \in & \span\setdef{\partial_{\vecx^\beta}(Q_A) \cdot \vecx^{\leq k\tau} \cdot Q_{\overline{A}} \cdot Q'_{\overline{B}}}{\begin{array}{c}\vecx^{\alpha} = \vecx^\beta \cdot \vecx^\gamma\;,\; A\subseteq [m]\;,\; \\ B \subseteq [r]\;,\; |A| + |B| = k\end{array}}\\
\implies \vecx^{= \ell}\partial_{\vecx^\alpha}(R) & \subseteq &  \span\setdef{\partial_{\vecx^\beta}(Q_A) \cdot \vecx^{\leq \ell + k\tau} \cdot Q_{\overline{A}} \cdot Q'_{\overline{B}}}{\begin{array}{c}\vecx^{\alpha} = \vecx^\beta \cdot \vecx^\gamma\;,\; A\subseteq [m]\;,\; \\ B \subseteq [r]\;,\; |A| + |B| = k\end{array}}\\
\implies \eval_S(\vecx^{= \ell}\partial_{\vecx^\alpha}(R)) & \subseteq &  \span\setdef{\eval_S\inparen{\partial_{\vecx^\beta}(Q_A) \cdot \vecx^{\leq \ell + k\tau} \cdot Q_{\overline{A}} \cdot Q'_{\overline{B}}}}{\begin{array}{c}\vecx^{\alpha} = \vecx^\beta \cdot \vecx^\gamma\;,\; A\subseteq [m]\;,\; \\ B \subseteq [r]\;,\; |A| + |B| = k\end{array}}
\end{eqnarray*}
If we focus on the term $\partial_{\vecx^\beta}(Q_A)$, since $Q_A$ is a product of $\SPS$ circuits of rank at most $\tau$, we have that $\partial_{\vecx^\beta}(Q_A)$ is a linear combination of terms $T_1 \cdots T_{|A|}$ where each $T_i$ is a product of linear polynomials and has rank at most $\tau$. Using \autoref{lem:low-rank-to-low-degree} on each of these $T_i$s, 
\[
\eval_S(\partial_{\vecx^\beta}(Q_A)) \in \span\inbrace{\eval_S(\vecx^{\leq (q-1)\tau |A|})}.
\]
Therefore,
\begin{eqnarray*}
\eval_S(\vecx^{= \ell}\partial_{\vecx^\alpha}(R)) & \subseteq &  \span\setdef{\eval_S\inparen{\partial_{\vecx^\beta}(Q_A) \cdot \vecx^{\leq \ell + k\tau} \cdot Q_{\overline{A}} \cdot Q'_{\overline{B}}}}{\begin{array}{c}\vecx^{\alpha} = \vecx^\beta \cdot \vecx^\gamma\;,\; A\subseteq [m]\;,\; \\ B \subseteq [r]\;,\; |A| + |B| = k\end{array}}\\
& \subseteq & \span\setdef{\eval_S\inparen{\vecx^{\leq \ell + k\tau + (q-1)k\tau} \cdot Q_{\overline{A}} \cdot Q'_{\overline{B}}}}{\begin{array}{c}\vecx^{\alpha} = \vecx^\beta \cdot \vecx^\gamma\;,\; A\subseteq [m]\;,\; \\ B \subseteq [r]\;,\; |A| + |B| = k\end{array}}.
\end{eqnarray*}
Finally, \autoref{lem:multilinearization} shows for every polynomial $f$, there is a multilinear polynomial of degree at most $\deg(f)$ that agrees with $f$ on all evaluations on a translate of a hypercube. Therefore, in the above span, we may assume that we are only shifting by multilinear monomials of degree $\ell + qk\tau$ as this doesn't change the evaluations $S \subseteq \vecc + \set{0,1}^n$. Hence,
\[
\eval_S(\vecx^{= \ell}\partial_{\vecx^\alpha}(R)) \subseteq  \span\setdef{\eval_S\inparen{\vecx^{\leq \ell + qk\tau}_{\mathrm{mult}} \cdot Q_{\overline{A}} \cdot Q'_{\overline{B}}}}{\begin{array}{c}\vecx^{\alpha} = \vecx^\beta \cdot \vecx^\gamma\;,\; A\subseteq [m]\;,\; \\ B \subseteq [r]\;,\; |A| + |B| = k\end{array}}.
\]
Therefore, using the fact that $m + r \leq (4d/\tau) + 1$, we get the bound
\[
\Gamma_{k,\ell,S}(R) := \dim\set{\eval_S(\vecx^{= \ell}\partial^{=k}(R))} \spaced{\leq} \binom{\frac{4d}{\tau} + 1}{k} \cdot \binom{n}{\ell + qk\tau} \cdot n  \]
where the first term corresponds to the number of choices for the subsets $A$ and $B$, and the last two terms correspond to the number of multilinear monomials of degree at most $\ell + qk\tau$. 
\end{proof}

\begin{remark}~\label{ref:remk-nonhomogeneous} \rm
Observe that, even if the circuit $C$ is of the form 
$$
C \spaced{=} \sum_a \prod_{b \in [m]} \sum_c \prod_d L_{abcd}
$$
 such that $\text{Size}(C) \leq 2^{\sqrt{d}/100}$ and $m = O(\frac{d}{\tau})$
, then the upper bound in \autoref{lem:upper-bound-circuit} continues to hold.\footnote{Essentially, in the proof of \autoref{clm:internal 2}, we already have an expression of the form $R^{\leq \tau} = Q_1 \cdots Q_m$ with $m = O\pfrac{d}{\tau}$ and the rest of the proof proceeds as expected.}
In particular, the formal degree of $C$ could be much larger than $d$ but if the product fan-in at level two of $C$ is small, then 
\[
\Gamma_{k,\ell,S}(C) \spaced{\leq} 2^{\sqrt{d}/100} \cdot \binom{O(\frac{d}{\tau})}{k} \cdot \binom{n}{\ell + k\tau q}\cdot \poly(n)\qedhere
\]
\end{remark}
\noindent
We later use this observation to complete the proof of \autoref{thm:mainthm-nonhom} in \autoref{sec:wrapping-proof}. 

\section{Lower bound for the complexity measure for an explicit polynomial}\label{sec:poly lower bound}

Let $\mathcal{E}$ be an arbitrary subset of $\F_q^n$ of size at most $\delta \cdot q^n$. We will choosing a specific set $S$ that shall be a subset of a translate of the hypercube and disjoint from $\mathcal{E}$. We will fix the precise definition of $S$ shortly once we motivate the requirements. 

The polynomial for which was shall prove our lower bound would be from the Nisan-Wigderson family. We would have to set our parameters carefully but for now we shall just be intentionally vague and refer to the polynomial as just $\NW$ and fix parameters at a later point. \\

\noindent
Associated with our measure $\Gamma_{k,\ell,S}(\NW)$ is a natural matrix that we shall call $\Lambda(\NW)$:
\begin{quote}
  The rows of $\Lambda(\NW)$ are indexed by a derivative $\partial_{\vecx^\alpha} \in \partial^{=k}$ of order $k$, and a monomial $\vecx^\beta$ of degree equal to $\ell$. 
The columns are indexed by all points $\veca \in S$. 
The entry in $(\vecx^{\beta} \cdot \partial_{\vecx^\alpha}, \veca)$ is the evaluation of $\vecx^{\beta} \cdot \partial_{\vecx^\alpha}(\NW)$ at the point $\veca$. 
\end{quote}
In other words, the matrix is just the vectors $\eval_S(\vecx^{\beta} \cdot  \partial_{\vecx^\alpha}(\NW))$ listed as different rows for each choice of $\vecx^{\alpha}$ and $\vecx^{\beta}$. Therefore, 
\begin{equation}~\label{eqn:gamma to rank bound}
{\Lambda(\NW)} = \Gamma_{k, \ell, S}(\NW)
\end{equation}
Recall from \autoref{lem:multilinearization} (multilinearization), as long as we only care about evaluations on a translate of a hypercube,
we can assume each row is the evaluations of the multilinearization of $\vecx^{\alpha}\cdot \partial_{\vecx^\beta}(\NW)$. This does not change the evaluation on any point $\veca \in S \subseteq \vecc + \mathcal{H}$. \\

Now any such matrix of evaluations can be naturally factorized as a coefficient matrix and an evaluation matrix. 
\begin{description}
\item{$C_{k,\ell}(\NW)$:} Each row is indexed by a derivative $\partial_{\vecx^{\alpha}}$ of order $k$, and a monomial $\vecx^\beta$ of degree $\ell$, and each column is indexed by a multilinear monomial $m$ of degree \emph{at most} $\ell + d - k$ over $n$ variables, and the $(\vecx^\beta \cdot \partial_{\vecx^\alpha}, m)$ entry is the coefficient of monomial $m$ in the multilinearization of $\vecx^\beta \cdot \partial_{\vecx^\alpha}(\NW)$ with respect to $\vecc + \mathcal{H}$ (\autoref{lem:multilinearization}).
\item{$V_t(S)$:}
Rows are indexed by multilinear monomials of degree at most $t = \ell + d - k$ over $n$ variables, columns are indexed by $\veca\in S$ and $(m,\veca)$ entry is the evaluation monomial $m$ at $\veca$. 
\end{description}

\noindent
Clearly, $\Lambda(\NW) = C_{k,\ell}(\NW) \cdot V_t(S)$. 
Thus if we can get a good lower bound on the ranks of the matrices $C_{k,\ell}(\NW)$ and $V_t(S)$ for a suitable set $S$, we would then be able to lower bound the rank of $\Lambda(\NW)$. 
This is formalized by the following simple linear algebraic fact. 

\begin{lemma}[Rank of products of matrices] \label{lem:rank-of-products}
If $A, B$ and $C$ are matrices such that $A = B \cdot C$, then $\rank(A) \geq \rank(B) + \rank(C) - (\text{\# rows of }C)$. 
\end{lemma}

\subsection{Rank of $C_{k,\ell}(\NW)$}\label{sec:rank-C}

Let us focus on the matrix $C_{k,\ell}(\NW)$ and restrict ourselves a submatrix $C'_{k,\ell}(\NW)$ to only those columns whose degree is \emph{exactly} $t = \ell + d - k$, and rows indexed by $(\vecx^\beta \cdot \partial_{\vecx^\alpha})$ with $\vecx^\beta$ being a multilinear monomial of degree exactly $\ell$.  

If our polynomial $\NW$ was multilinear, if we were to read off any row of $C'_{k,\ell}(\NW)$, this is just the list of coefficients of all multilinear monomials of $(\vecx^\beta \cdot \partial_{\vecx^\alpha}(\NW))$. This is because the multilinearization operation in \autoref{lem:multilinearization} maps any non-multilinear monomial to a multilinear polynomial of strictly smaller degree and hence these monomials are not included in the columns of $C'_{k,\ell}$. 

The key point here is that the matrix $C'_{k,\ell}(\NW)$ is just the matrix of  \emph{projected shifted partial derivatives} of $\NW$. The results of Kayal et. al~\cite{KLSS} and Kumar and Saraf~\cite{KS14} give a lower bound on the rank of this matrix for a suitable choice of the polynomial, but the lower bound of Kumar and Saraf~\cite{KS14} is more relevant as it is true over any field (unlike \cite{KLSS} that works only over characteristic zero fields). 

Using a tight analysis of the argument in \cite{KS14}, that we present in \autoref{sec:tight-lb-proof} we obtain the following lemma for the Nisan-Wigderson polynomial for very carefully chosen parameters.

\begin{lemma}[Tight analysis of \cite{KS14}]\label{lem:KS-tight-bound}
For every $d$ and $k = O(\sqrt{d})$ there exists parameters $m,e,  \epsilon$ such that $m = \Theta(d^2)$, $n = md$ and $\epsilon = \Theta\pfrac{\log d}{\sqrt{d}}$ with
\begin{eqnarray*}
  m^k & \geq & (1+\epsilon)^{2(d-k)}\\
  m^{e-k} & = & \pfrac{2}{1+\epsilon}^{d-k} \cdot \poly(m). 
\end{eqnarray*}
For any $\set{d,m,e,k,\epsilon}$ satisfying the above constraints,  the polynomial $\NW_{d,m,e}$, if $\ell = \frac{n}{2}(1 - \epsilon)$, then over any field $\F$, we have
\[
\rank(C_{k,\ell}(\NW_{d,m,e})) \spaced{\geq} \rank(C'_{k,\ell}(\NW_{d,m,e})) \spaced{\geq} \binom{n}{\ell + d - k} \cdot \exp(-O(\log^2 d)).
\]
\end{lemma}

\subsection{Rank of $V_t(S)$}

Let $\mathcal{H}_{\leq t}$ refer to elements of $\set{0,1}^n$ of Hamming weight at most $t$. Our first step would be to choose our set $S$ carefully so that we can maximize the rank of $V_t(S)$. 

\begin{observation}\label{obs:redn-to-hypercube}
  Let $\mathcal{E}$ be a subset of $\F_q^n$ of size at most $\delta \cdot q^n$. 
  Then for any $0\leq t \leq n$, there is a vector $\vecc \in \F_q^n$ such that 
  \[
  \abs{(\vecc + \mathcal{H}_{\leq t}) \intersection \mathcal{E}} \spaced{\leq} \delta \cdot \abs{H_{\leq t}}. 
  \]
\end{observation}
\begin{proof}
  Let $\mathbb{1}_\mathcal{E}(\veca)$ be the indicator function that is $1$ if $y\in \mathcal{E}$, and $0$ otherwise. 
  Then, 
  \[
  \delta \spaced{\geq} \E_{y\in \F_q^n}\insquare{\mathbb{1}_\mathcal{E}(\veca)} \spaced{=} \E_{c\in \F_q^n}\insquare{\E_{y\in \mathcal{H}_{\leq t}}\insquare{\mathbb{1}_\mathcal{E}(\vecc + \veca)}}.
  \]
  Thus, there exists some $\vecc\in \F_q^n$ that still maintains the inequality. 
\end{proof}

\noindent
Our set would be $S = (\vecc + \mathcal{H}_{\leq t}) \setminus \mathcal{E}$, which has the property that $\abs{S \intersection (\vecc + \mathcal{H}_{\leq t})} \geq (1 - \delta) \cdot \abs{\mathcal{H}_{\leq t}}$ by the above observation, and $S \intersection \mathcal{E} = \emptyset$. 

Let $V_t(S-\vecc)$ be the matrix whose rows are indexed by the polynomials $m(\vecx - \vecc)$, where $m$ is a multilinear monomials in variables $\vecx$ of degree at most $t$. The columns of $V_t(S-\vecc)$ are indexed by $S$. We have the following observation.
\begin{observation}\label{obs:rank-unchanged-by-shifting}
$\rank(V_t(S)) = \rank(V_t(S - \vecc))$. 
\end{observation}
\begin{proof}
For any multilinear monomial $m$ of degree at most $t$, the polynomial $m(\vecx - \vecc)$  is multilinear and has degree at most $t$.
Thus clearly, the row-space of $V_t(S-\vecc)$ is contained in the row-space of $V_t(S)$. 
The converse also holds trivially as the translation is invertible. 
\end{proof}

We now prove our next lemma which shows a lower bound on the rank of $V_t(S-\vecc)$.
\begin{lemma}~\label{lem:rank of shifted matrix}
For any set $S \subseteq \set{0,1}^n \subset \F_q^n$ and any $0\leq t \leq n$, 
\[
\rank(V_t(S-\vecc)) \spaced{=} \abs{S}
\]
\end{lemma}
\begin{proof}
Since $S \subseteq \vecc + \mathcal{H}_{\leq t}$, the set $S' := S - \vecc \subset \mathcal{H}_{\leq t}$. Thus the matrix $V_t(S - \vecc)$ is simply the matrix $V_t(S')$. For any $\veca \in \set{0,1}^n$, let $m_\veca$ refer to the `characteristic' monomial $\prod_{i:a_i =1} x_i$, and let $m_{\mathbf{0}} = 1$.  \\

Consider the sub-matrix of $V_t(S')$ by restricting to rows indexed by $\setdef{m_{\veca}}{\veca \in S'}$.
By rearranging the rows and columns in the increasing order of the weight of $\veca$, it is evident that the sub-matrix is upper-triangular with ones on the diagonal. 
It therefore follows that the rank of $V_t(S')$ (which is just $V_t(S - \vecc)$) is at least $\abs{S'} = \abs{S}$. 
\end{proof}

\noindent
Combining \autoref{obs:rank-unchanged-by-shifting} and \autoref{lem:rank of shifted matrix}, we have our required bound on the rank of $V_t(S)$. 

\begin{lemma}\label{lem:rank-Vt}
Let $\mathcal{E}$ be an arbitrary subset of $\F_q^n$ of size at most $\delta \cdot q^n$. 
Then, there exists a set $S \subset \F_q^n \setminus \mathcal{E}$ such that $S \subseteq \vecc + \mathcal{H}$ for some $\vecc \in \F_q^n$ for which
\[
\rank(V_t(S)) \spaced{\geq} (1- \delta) \cdot \abs{\mathcal{H}_{\leq t}} \spaced{=} (1 - \delta) \cdot (\text{\# rows of $V_t(S)$})
\]
\end{lemma}

\subsection*{Putting them together}

\begin{lemma}[Rank bound for $\Lambda(\NW_{d,m,e})$]~\label{lem:rank bound for nw}
Let $\mathcal{E}$ be an arbitrary subset of $\F_q^n$ of size at most $\delta \cdot q^n$, with $\delta = \exp(-\omega(\log^2 d))$. 
Then, there exists a set $S \subset \F_q^n \setminus \mathcal{E}$ such that $S \subseteq \vecc + \mathcal{H}$ for some $\vecc \in \F_q^n$ for which
\[
\rank(\Lambda(\NW_{d,m,e})) \spaced{\geq} \binom{n}{\ell + d - k} \cdot\exp(-O(\log^2 d))
\]
where the parameters $d,m,e,k,\ell$ are chosen as in \autoref{lem:KS-tight-bound}. 
\end{lemma}
\begin{proof}
  Consider the set $S$ chosen in \autoref{lem:rank-Vt} (for $t = \ell + d - k$). By \autoref{lem:rank-Vt},
\[
\rank(V_t(S)) - (\text{\# rows of $V_t(S)$}) \spaced{\leq} (- \delta) \abs{\mathcal{H}_{\leq t}} \spaced{\leq} (- \delta) \cdot n \cdot \binom{n}{\ell + d - k}
\]
\autoref{lem:KS-tight-bound} shows that rank of $C_{k,\ell}(\NW_{d,m,e})$ can be lower bounded by 
\[
\rank(C_{k,\ell}(\NW_{d,m,e})) \spaced{\geq}  \binom{n}{\ell + d - k} \cdot \exp(-O(\log^2 d))
\]
Thus, since $\Lambda(\NW_{d,m,e}) = C_{k,\ell}(\NW_{d,m,e}) \cdot V_t(S)$ with $t = \ell + d - k$, \autoref{lem:rank-of-products} implies that
\begin{eqnarray*}
\rank(\Lambda(\NW_{d,m,e})) & \geq & \rank(C_{k,\ell}(\NW_{d,m,e})) + \rank(V_t(S)) - (\text{\# rows of $V_t(S)$})\\
& \geq & \binom{n}{\ell + d - k} \cdot \exp(-O(\log^2 d)) - \delta \cdot n \cdot \binom{n}{\ell + d - k}\\
& \geq & \binom{n}{\ell + d - k} \cdot \exp(-O(\log^2 d))\quad\quad \text{as $\delta = \exp(-\omega(\log^2 d))$.}\qedhere
\end{eqnarray*}
\end{proof}

Combining this with \autoref{eqn:gamma to rank bound}, we get the following lemma.
\begin{lemma}[Rank bound for $\Lambda(\NW_{d,m,e})$]~\label{lem:complexity bound for nw}
Let $\mathcal{E}$ be an arbitrary subset of $\F_q^n$ of size at most $\delta \cdot q^n$, with $\delta = \exp(-\omega(\log^2 d))$. 
Then, there exists a set $S \subset \F_q^n \setminus \mathcal{E}$ such that $S \subseteq \vecc + \mathcal{H}$ for some $\vecc \in \F_q^n$ for which
\[
\rank(\Gamma_{k, \ell, S}(\NW_{d,m,e})) \spaced{\geq} \binom{n}{\ell + d - k} \cdot\exp(-O(\log^2 d))
\]
\end{lemma}
\section{Wrapping up the proof}\label{sec:wrapping-proof}

\begin{theorem}\label{thm:mainthm}
Let $\F_q$ be the finite field of cardinality $q$. Let $C$ be a depth-$5$ circuit of formal degree at most $2d$ which computes the polynomial $\NW_{d,m,e}$ with parameters as in \autoref{lem:KS-tight-bound}. Then
 $$\text{Size}(C) > 2^{\sqrt{d}/100}.  $$
Further, the same lower bound holds even if $C$ was a circuit of the form 
\[
C \spaced{=} \sum_i \prod_{j \in [m]} \sum_k \prod_\ell L_{ijk\ell}
\]
with $m = O(\sqrt{d})$. 
\end{theorem}
\begin{proof}
We shall prove the above theorem for homogeneous depth-$5$ circuits. The lower bound for such non-homogeneous circuits would also follow directly since such circuits also have the same upper-bound on the complexity measure (\autoref{ref:remk-nonhomogeneous}).\\

Assume on the contrary that there is a circuit $C$ computing $\NW_{d,m,e}$ with $\text{Size}(C) \leq 2^{\sqrt{d}/100}$. 
Let $\tau$ be as defined in \autoref{cor:corollarywp} and let $k = \tau/2q^3$. 
Let $\mathcal{E} = \mathcal{E(C)}$ be the set as defined in \autoref{subsec:upper-bound}. We know that  
$$|\mathcal{E}| \leq \delta\cdot q^n$$
for some $\delta = \exp(-O(\sqrt{d}))$. Let $\ell = \frac{n}{2}\inparen{1 - \epsilon}$ where $\epsilon = \frac{\log d}{c \sqrt{d}}$ chosen as in \autoref{lem:KS-tight-bound}. Since $n = d^3$, clearly we have satisfy $\ell + k\tau q< n/2$. 
Let  $S \subset \F_q^n \setminus \mathcal{E}$ be the set guaranteed by \autoref{lem:complexity bound for nw}. From \autoref{lem:complexity bound for nw}, we know that 
$$\Gamma_{k, \ell, S}(\NW) \geq \binom{n}{\ell + d - k} \cdot\exp(-O(\log^2 d)) $$
This may be simplified using \autoref{lem:binom-approx} to
$$ 
\Gamma_{k, \ell, S}(\NW) \geq  \binom{n}{\ell} \cdot (1+\epsilon)^{2d-2k} \cdot \exp(-O(d\epsilon^2)) \cdot \exp(-O(\log^2 d)) 
$$
Also, from \autoref{lem:upper-bound-circuit}, we know that
$$\Gamma_{k, \ell, S}(C) \leq 2^{\sqrt{d}/100}\cdot \binom{\frac{4d}{\tau} + 1}{k} \cdot \binom{n}{\ell + qk\tau} \cdot \poly(n)$$
Again, using \autoref{lem:binom-approx}, we get 
$$\Gamma_{k, \ell, S}(C) \leq 2^{\sqrt{d}/100}\cdot \binom{\frac{4d}{\tau} + 1}{k} \cdot \binom{n}{\ell} \cdot (1+\epsilon)^{2qk\tau} \cdot \exp(O(-qk\tau\cdot \epsilon^2)) \cdot \poly(n)$$
Since $C$ computes $\NW_{d,m,e}$, so it must be the case that
\begin{eqnarray*}
2^{\sqrt{d}/100}\cdot \poly(n) & \geq & (1+\epsilon)^{(d-k) + (d-k -2qk\tau)}\cdot \exp(-O_q(\log^2 d))
\end{eqnarray*}
Since $k = \tau/2q^3$, so $2qk\tau = \tau^2/q^2$. From our choice of $\tau$ in \autoref{cor:corollarywp}, $\tau < \frac{q\sqrt{d}}{6}$. So $$2qk\tau = \tau^2/q^2 < d/36$$ 
Therefore, 
\begin{eqnarray*}
2^{\sqrt{d}/100}\cdot \poly(n) \geq (1+\epsilon)^{(d-k)}\cdot \exp(-O_q(\log^2 d))
\end{eqnarray*}
But this is a contradiction since $(1+\epsilon)^{(d-k)} = \exp({\Omega(\sqrt{d}\log d)})$ by our choice of parameters. Therefore, the size of $C$ is at least $2^{\sqrt{d}/100}$.
\end{proof}

In fact, the above proof gives more. It shows that if we have a depth-$5$ circuit computing $\NW_{d,m,e}$ over $\F_q$, then either the number of high-rank terms is at least $2^{\sqrt{d}/50}$ or the top fan-in is $\exp(\Omega(\sqrt{d}\log d))$. 

\subsection{Getting the right order of quantifiers}\label{sec:order-of-quantifiers}

In our proof so far, we first fix the field $\F_q$ that we are working over and the parameters of our polynomial are then chosen based on $q$. Thus, as $q$ varies, the polynomial for which we show the lower bound also seems to vary. The ideal scenario would be to construct a fixed polynomial family so that for every $q$ we get a lower bound of $\exp(\Omega_q(\sqrt{d}))$. We do that now, and this would complete the proof of \autoref{thm:mainthm-all fields}. 

We now show that this is not necessary. We would fix a family of polynomials as defined in \autoref{def:NW final}, and then argue that for every finite field $\F_q$ (here we think of $q$ as $O(1)$), there is a projection of this family, obtained by setting some of the variables to $0$ or $1$, such that \autoref{thm:mainthm} holds. This would complete the proof of \autoref{thm:mainthm-all fields}. \\

We shall be dealing with a lot of parameters and constraints. The following is essentially the ``zone'' in which we can prove strong lower bounds (\autoref{lem:KS-tight-bound}). 

\begin{definition}[Goldilocks Zone]\label{defn:goldilocks-zone}
We shall say that parameters $m,d,e, k,\epsilon$ with $k = \Theta(\sqrt{d})$ and $\epsilon = \Theta\pfrac{\log d}{\sqrt{d}}$ lie in the \emph{Goldilocks Zone} if they satisfy
\begin{eqnarray*}
 m^k & \geq & (1+\epsilon)^{2(d-k)}\\
 m^{e-k} & = & \pfrac{2}{1+\epsilon}^{d-k} \cdot \poly(m).
\end{eqnarray*}
\end{definition}

Recall that for \autoref{lem:KS-tight-bound}, and consequently \autoref{thm:mainthm}, the parameters $m, d, e, k$ must lie in the Goldilocks zone. The crucial point is that this is a field dependent condition since $k$ (and hence everything else) explicitly depends on $q$. In the next lemma, we show that we can start with a fixed polynomial  such that for every finite field $\F_q$ of fixed size, there exists a $0, 1$ projection which lies in the Goldilocks zone.  
\begin{lemma}\label{lem:goldilocks-projection}
Consider the $\NW_{d,m,e}$ polynomial with $m = \Theta(d^2)$ and $e$ chosen so that
\[
m^{e} \spaced{=} 2^{d}\cdot \poly(m).
\]
Suppose $k = \Theta(\sqrt{d})$ and $\epsilon = \Theta\pfrac{\log d}{\sqrt{d}}$ satisfy the constraint $m^k > (1+\epsilon)^{2(d-k)}$. Then, there exists a $d' \in [d-O(\sqrt{d}\log d),d]$ such that $\NW_{d',m,e}$ is a $0/1$ projection of $\NW_{d,m,e}$ and the parameters $\set{d',m,e,k,\epsilon}$ fall in the Goldilocks Zone. 
\end{lemma}
\begin{proof}
Since $m^e = 2^d\cdot \poly(m)$, $m^k > (1+\epsilon)^{2(d-k)} \text{ and } (1+\epsilon)^{d} = \exp(\Theta(\sqrt{d}\log d))$, we have
\[
m^{e-k} \spaced{=}  \pfrac{2}{1+\epsilon}^{d-k}\cdot \exp(-\Theta(\sqrt{d}\log d)).
\]
Since the slack in $m^{e-k}$ is just $\exp(\Theta(\sqrt{d}\log d))$ (when compared to the desired value in \autoref{defn:goldilocks-zone}), there exists some $d' \in [d - O(\sqrt{d}\log d),d]$ such that 
\[
m^{e-k} \spaced{=} \pfrac{2}{1+\epsilon}^{d' - k} \cdot \poly(m).
\]
Further, since $m^k > (1+\epsilon)^{2(d-k)}$, it follows that $m^k > (1+\epsilon)^{2(d'- k)}$ as $d' < d$. Hence the parameters $\set{d',m,e,k,\epsilon}$ indeed fall in the Goldilocks Zone (~\autoref{defn:goldilocks-zone}). 

It suffices to show that $\NW_{d',m,e}$ is a projection of $\NW_{d,m,e}$. This is readily seen as setting the variables $x_{ij} = 1$ for all $i \in [d - d']$ and $j \in [m]$ yields $\NW_{d',m,e}$ up to relabelling variables. 
\end{proof}

With this, we can finally prove our main theorems. 

\begin{theorem}[\autoref{thm:mainthm-all fields} restated]
Consider the polynomial $\NW_{d,m,e}$ with parameters chosen such that $m = \Theta(d^2)$ and $m^e = 2^d \cdot \poly(m)$. Then, for any fixed finite field $\F_q$, any homogeneous depth-$5$ circuit over $\F_q$ computing $\NW_{d,m,e}$ must have size at least $2^{\sqrt{d}/200}$. 
\end{theorem}
\begin{proof}
Fix a field $\F_q$ and let $k = \sqrt{d}/12q^3$. 

Suppose on the contrary that there is indeed a homogeneous depth-$5$
circuit $C$ computing $\NW_{d,m,e}$. Then, by \autoref{lem:goldilocks-projection}, this also implies there is
a projection $C'$ that computes the $\NW_{d',m,e}$ such that there is an $d - O(\sqrt{d}\log d) \leq d' \leq d$ and there is an $\epsilon = \Theta\pfrac{\log d}{\sqrt{d}}$ for which $\set{d',m,e,k,\epsilon}$ fall in the Goldilocks Zone (\autoref{defn:goldilocks-zone}). Now $C'$ is a circuit formal degree $d \leq d' + O(\sqrt{d}\log d) \leq 2d'$ that computes the polynomial $\NW_{d',m,e}$. By \autoref{thm:mainthm}, this implies that 
\[
\size(C) \geq \size(C') > 2^{\sqrt{d'}/100} > 2^{\sqrt{d}/200}. \qedhere
\]
\end{proof}

\noindent
The proof of this theorem also follows along the same lines. 

\begin{theoremwp}[~\autoref{thm:mainthm-nonhom} restated]
Consider the polynomial $\NW_{d,m,e}$ with parameters chosen such that $m = \Theta(d^2)$ and $m^e = 2^d \cdot \poly(m)$. Then, for any fixed finite field $\F_q$, any depth-$5$ circuit over $\F_q$ of the form 
\[
C \spaced{=} \sum_i \prod_{j \in [m]} \sum_k \prod_\ell L_{ijk\ell}
\]
where each $L_{ijk\ell}$ is a linear polynomial and $m = O(\sqrt{d})$ that computes $\NW_{d,m,e}$ must have size at least $2^{\sqrt{d}/200}$. 
\end{theoremwp}





\section{Discussion}

\subsection{Connections between arithmetic circuits over $\F_q$ and $\mathsf{AC}^0[\bmod q]$}\label{sec:AC0-connections}

Although constant depth arithmetic circuits over $\F_q$ appear to be similar to the class $\mathsf{AC}^0[\bmod q]$, they are surprisingly very different with respect to functions computed by them. A striking example, due to Agrawal, Allender and Datta~\cite{aad00}, is that arithmetic circuits over $\F_3$ can ``compute'' both the $\mathrm{Mod}3$ function, as well as the $\mathrm{Mod}2$ function via
\[
\mathrm{Mod}2(x_1,\dots, x_n) \quad=\quad \inparen{2 + \prod_{i=1}^n (1 + x_i)}^2.
\] 
However, it is true that functions computed by arithmetic circuits over $\F_{p^k}$ have strong connections with $\mathsf{AC}^0[\bmod p(p^k - 1)]$ but unless we are working over $\F_2$ it seems difficult lift a lower bound for $\mathrm{AC}^{0}[\bmod p]$ to arithmetic circuits over $\F_p$. For more on this, see \cite{aad00}. 

The only exception we know of is the result of Grigoriev and
Razborov~\cite{gr00} where they lift Smolensky's \cite{smolensky87}
lower bound for $\mathrm{AC}^{0}[\bmod p]$ to depth-$3$ arithmetic
circuits over $\F_p$, and this crucially uses the fact that depth-$3$
arithmetic circuits can be point-wise approximated by a ``sparse
polynomial''. But in general, constant depth arithmetic circuits over $\F_p$ and
boolean circuits in $\mathrm{AC}^{0}[\bmod p]$ seem to be two very
different classes. 

\subsection{Lower bounds for iterated matrix multiplication}
Given the results in this paper, one might wonder if the lower bounds in this paper work for a polynomial in $\VP$. One natural candidate polynomial for which one might hope to show such a lower bound would be the iterated matrix multiplication polynomial (IMM). It was shown in \cite{KS14} that IMM has a large complexity with respect to the measure of projected shifted partial derivatives. Unfortunately, the bounds in \cite{KS14} only show that the dimension of the space of projected shifted partial derivatives of the IMM (degree $d$ in $d^{O(1)}$ variables) are $\exp{(-\delta\sqrt{d}\log d)}$ factor close to the maximum possible value for some constant $\delta$. This slack seems  to be insufficient for the proofs in this paper to work as in the proof of \autoref{lem:complexity bound for nw}, we relied on the fact that for the polynomial $\NW$, the projected shifted partials complexity was at most a quasi-polynomial factor away from the largest possible. 

\subsection{Finer separations for bounded depth circuits ?}
In \cite{KS14a}, it was shown that homogeneous depth-$4$ circuits are exponentially more powerful than homogeneous depth-$4$ circuits with bounded bottom fan-in. A natural question to ask is  whether such separations can be shown between homogeneous depth-$4$ circuit and homogeneous depth-$5$ circuits. One of the first strategies to attempt for this question would be to try and show that there is a homogeneous depth-$5$ circuit such that its projected shifted partial derivative complexity is quite large. The results in this paper show that the measure can not to be too close to the largest possible value, in particular it needs to be at least a factor $\exp(-\Omega(\sqrt{d}))$ away from the largest possible value. If this is bound is tight, then such a separation between homogeneous depth-$5$ circuits and  homogeneous depth-$4$ circuits can still be shown using projected shifted partial derivatives. However, it is not clear  if this is the case. As mentioned before, even the known lower bounds on the dimension of the projected shifted partials for the IMM seem a factor $\exp{(-\Omega(\sqrt{d}\log d))}$ away from the largest possible value. 
 
 \subsection{The tightness of the results and relevance to $\VP$ vs $\VNP$}
For homogeneous depth-$4$ circuits, we know $\exp{(\Omega(\sqrt{d}\log d))}$ lower bounds~\cite{KLSS, KS14} and any asymptotic improvement in the exponent would imply general arithmetic circuit lower bounds. In this sense, the lower bounds for homogeneous depth-$4$ circuits are tight, over all fields. It is natural to ask, if the bounds in this paper are tight in this sense? The answer to this question is far from obvious to us. In particular, it is not clear if we can use computational advantage of having  linear forms at the bottom level of the circuit to get a better depth reduction from $\VP$ to homogeneous depth-$4$ circuits, when compared to depth reduction to homogeneous depth-$4$ circuits.

\subsection{Lower bounds over fields of characteristic zero ?}\label{sec:lb-char-zero}
One might wonder if the techniques in this paper could be potentially adapted to work for depth-$5$ circuits over fields of characteristic zero. In the proof of \autoref{lem:upper-bound-circuit}, we strongly relied on the fact that we are working over a fixed finite field, so it clearly seems hard to generalize over large fields (even when the characteristic is small). In addition to this obvious technical obstruction to generalizing the proof in this paper, there seems to be another reason why the proof strategy in this paper could be hard to replicate over fields of characteristic zero, namely, an analog of \autoref{thm:mainthm-nonhom} over fields of characteristic zero would imply that $\VP \neq \VNP$. The reason is that over characteristic zero fields, one can obtain better depth reductions to non-homogeneous depth-$5$ circuits by combining \cite{av08,koiran,Tav13} with \cite{gkks13b}. Although this is reasonably well known, we give a formal proof here for completeness. 

The following lemma is a simple generalization of the proof of depth reduction to depth-$4$ circuits by Tavenas~\cite{Tav13}. 
\begin{lemma}[Depth reduction to homogeneous depth six circuits]~\label{lem:depth 6 reduction}
Let $P$ be a polynomial of degree $d$ in $\poly(d)$ variables which can be computed by an arithmetic circuit $C$ of size $\poly(d)$. Then, there is a homogeneous depth-6 circuit $C'$ which computes $P$ and satisfies 
\begin{itemize}
\item $\text{Size(C)} \leq \exp{(O(d^{1/3}\log d))}$, and
\item The fan-in of all the product gates in $C'$ is bounded by $O(d^{1/3})$.
\end{itemize}
\end{lemma}  
Now, we start with the circuit $C'$ as guaranteed by the lemma above, and for each of the product gates at the second level, look at its inputs. Each such input is a $\sum\prod^{O(d^{1/3})}\sum\prod^{O(d^{1/3})}$ circuit (depth-$4$ circuit with all product fan-ins being at most $O(d^{1/3})$) of size at most $\exp{(O(d^{1/3}\log d))}$.  We now apply the depth reduction to depth-$3$ by Gupta et al.~\cite{gkks13b}
to each one of these depth-$4$ circuits. As a result, each of these depth-$4$ circuits get reduced to a depth-$3$ circuit, with at most a factor of $\exp{(O(d^{1/3}))}$ blow up in size. Plugging these depth-$3$ circuits back into $C'$, we obtain a depth-$5$ circuit $C''$ such that 
\begin{itemize}
\item $\text{Size(C)} \leq \exp{(O(d^{1/3}\log d))}$, and
\item The fan-in of all the product gates at level two of $C''$ is bounded by $O(d^{1/3})$.
\end{itemize} 
Recall that the depth reduction in\cite{gkks13b} only works over fields of characteristic zero. 
This yields the following depth reduction to non-homogeneous depth-$5$ circuits. 
\begin{lemma}[Depth reduction to non-homogeneous depth-$5$ circuits]~\label{lem:depth 5 reduction}
Let $\F$ be a field of characteristic zero. Let $P$ be a polynomial of degree $d$ in $\poly(d)$ variables over $\F$ which can be computed by an arithmetic circuit $C$ of size $\poly(d)$. Then, there is a depth-5 circuit $C''$ which computes $P$ and satisfies 
\begin{itemize}
\item $\text{Size(C)} \leq \exp{(O(d^{1/3}\log d))}$, and
\item The fan-in of all the product gates at level two of $C'$ is bounded by $O(d^{1/3})$.
\end{itemize}
\end{lemma}  

Now, observe that an analogue of \autoref{thm:mainthm-nonhom} over fields of characteristic zero, would imply an $\exp{(\Omega(d^{1/2}))}$ lower bound
for the depth-$5$ circuits obtained in \autoref{lem:depth 5 reduction}, and hence imply $\VP \neq \VNP$.  
\section*{Acknowledgements}
We are very  grateful to Mike Saks and Shubhangi Saraf for many discussions and much encouragement. The chat with Mike about powering circuits over finite fields was specially insightful. Part of this work was done while the first author was an intern at MSR New England. We are thankful to Madhu Sudan and the other members of the lab, for  stimulating discussions and  generous hospitality. Many thanks to Madhu for patiently sitting through a presentation of the proof. 

We would also like to thank Eric Allender for answering our questions about connections between boolean circuits and arithmetic circuits  over finite fields and pointing us to reference~\cite{aad00}.

This work was partly motivated by discussions over the questions of the projected shifted partials complexity of homogeneous depth-$5$ circuits while the authors were at MSR Bangalore in Summer'14. We are grateful to Neeraj Kayal for hosting us and for many insightful conversations. 

\bibliographystyle{RPurl}

\bibliography{references}

\appendix

\section{Tight analysis of the \cite{KS14} lower bound}\label{sec:tight-lb-proof}

\def\GammaPSD{\Gamma^{\mathrm{PSD}}}
\def\mult{\mathrm{mult}}
\def\LM{\mathrm{LM}}

We recall the measure of \emph{projected shifted partial derivatives} that was used in \cite{KLSS} and \cite{KS14}. 
\[
\GammaPSD_{k,\ell}(P) \spaced{=} \dim\set{\mult\inparen{\vecx^{=\ell} \partial^{=k}(P)} }
\]
where $\mult(f)$ is just the polynomial $f$ restricted to just its multilinear monomials. 
As mentioned before, this $\GammaPSD_{k,\ell}(P)$ is precisely $\rank(C'_{k,\ell}(P))$ as defined in \autoref{sec:rank-C}. \\

\noindent 
The goal of this section would be to prove \autoref{lem:KS-tight-bound} that we restate below. 

\begin{lemma*} For every $d$ and $k = O(\sqrt{d})$ there exists parameters $m,e,  \epsilon$ such that $m = \Theta(d^2)$, $n = md$ and $\epsilon = \Theta\pfrac{\log d}{\sqrt{d}}$ with
\begin{eqnarray*}
  m^k & \geq & (1+\epsilon)^{2(d-k)}\\
  m^{e-k} & = & \pfrac{2}{1+\epsilon}^{d-k} \cdot \poly(m). 
\end{eqnarray*}
For any $\set{d,m,e,k,\epsilon}$ satisfying the above constraints,  the polynomial $\NW_{d,m,e}$, if $\ell = \frac{n}{2}(1 - \epsilon)$, then over any field $\F$, we have
\[
\GammaPSD_{k,\ell}(\NW_{d,m,e}) \spaced{\geq} \binom{n}{\ell + d - k} \cdot \exp(-O(\log^2 d)).
\]
\end{lemma*}

\noindent 
The rest of this section would just be a proof of this lemma. \\

Before we proceed to the lower bound on $\GammaPSD_{k,\ell}(\NW_{d,m,e})$, let us first show that we can indeed find parameters that satisfy the above constraints. 
Fix $m$ to be the smallest power of $2$ greater than $d^2$ to get $m = \Theta(d^2)$. Next, we shall fix the constant $c$ in $\epsilon = \frac{\log d}{c \sqrt{d}}$ so that 
\[
m^k \spaced{\geq} (1+\epsilon)^{2(d-k)}
\]
This is always possible by choosing $c$ to be large enough as $(1+\epsilon)^{d-k} = \exp(O(\sqrt{d}\log d))$ and that is also the order of $m^k$. 

Once we have done that, we shall fix $e$ so as to ensure that 
\[
m^{e-k} \spaced{=} \pfrac{2}{1+\epsilon}^{d-k} \cdot \poly(m)
\]
This is always possible because choosing $e = k$ makes the $\mathrm{LHS}< \mathrm{RHS}$ and choosing $e = m$ makes $\mathrm{LHS} > \mathrm{RHS}$. Hence, there must be an integer $e$ such that $\mathrm{LHS}$ and $\mathrm{RHS}$ are within a multiplicative factor of $m$. 

\bigskip

All lower bounds on  the dimension of shifted partial derivatives of a polynomial $P$  was obtained by finding a \emph{large} set of \emph{distinct leading monomials}. 
In \cite{KS14}, they take this approach but require a very careful analysis. 
The key difference in this setting is the following: 

\begin{quote}
  If $\beta$ is the leading monomial of a polynomial $P$, then for any monomial $\gamma$, we also have that $\beta \cdot \gamma$ is the leading monomial of $\gamma P$. 

  However, the leading monomial of $\mathrm{mult}(\gamma P)$ could be $\beta' \cdot \gamma$ for some $\beta' \neq \beta$ (as higher monomials could be made non-multilinear during the shift by $\gamma$). 
\end{quote}

The multilinear projection makes the task of counting leading monomials much harder and \cite{KS14} come up an indirect way to count them. Throughout this discussion, let $\LM(f)$ refer to the leading monomial of $f$ in some natural ordering, say the lexicographic order. 

\subsubsection*{Leading monomials after multilinear projections}

Let $P$ the polynomial of degree $d$ for which we are trying to lower bound $\Gamma^{\mathrm{PSD}}_{k,\ell}(P)$. 
For every monomial multilinear monomial $\alpha$ of degree $k$, and a monomial $\beta \in \partial_\alpha(P)$, define the set $A(\alpha, \beta)$ as
\[
A(\alpha, \beta) \spaced{=} \setdef{\gamma}{\begin{array}{c}\deg(\gamma) = \ell + d - k\;\text{and there is a $\gamma'$ of degree $\ell$}\\\text{such that }\gamma  = \mathrm{LM}(\mathrm{mult}(\gamma' \cdot \partial_\alpha(P))) = \gamma' \cdot \beta \end{array}}
\]
In other words, we want the number of distinct monomials that are contributed by $\beta$, which are also distinct leading monomials obtained from $\partial_\alpha(P)$ that are divisible by $\beta$. 
We then have
\[
\Gamma^{\mathrm{PSD}}_{k,\ell}(P) \spaced{\geq} \abs{\Union_{\alpha, \beta} A(\alpha, \beta)}
\]

\noindent
{\bf Choice of derivatives:} Instead of looking at all derivatives in  $\partial^{=k}$, we shall restrict ourselves to just a subset of derivatives. Restricting the above union to a subset $\Delta  \subset \vecx^{=k}$ still continues to remain a lower bound for $\GammaPSD_{k,\ell}(P)$. Keeping in mind that we are dealing with $P = \NW_{d,m,e}$ and that $m^k > (1+\epsilon)^{2(d-k)}$. We shall choose $\Delta$ to be a set of size exactly $(1+\epsilon)^{2(d-k)}$ which consists of monomials of the form $x_{1a_1}\cdots x_{ka_k}$ with each $a_i \leq m$. This shall become relevant later. 
\begin{equation}\label{eqn:union-of-As}
\Gamma^{\mathrm{PSD}}_{k,\ell}(P) \spaced{\geq} \abs{\Union_{\substack{\alpha \in \Delta \\\beta \in \vecx^{=\ell}}} A(\alpha, \beta)}
\end{equation}

\bigskip

We shall need the following lemma from \cite{KS14} that is a strengthening of the standard Inclusion-Exclusion principle. 

\begin{lemma}[Stronger Inclusion-Exclusion \cite{KS14}]\label{lem:str-inc-exc}
Let $A_1,\dots, A_r$ be sets such that there is some $\lambda > 1$ such that
\[
\sum_{i\neq j} \abs{A_i \intersection A_j} \spaced{\leq} \sum_i \lambda \cdot \abs{A_i}
\]
Then, 
\[
\abs{\Union_i A_i} \spaced{\geq} \inparen{\frac{1}{4\lambda}} \cdot \inparen{\sum_i \abs{A_i}}
\]
\end{lemma}
\begin{corollary}\label{cor:inc-exc-str}
Considers sets $A_1,\dots, A_r$  and let $S_1 = \sum_i \abs{A_i}$ and $S_2 = \sum_{i\neq j} \abs{A_i \intersection A_j}$. 
Then, 
\[
\abs{\Union A_i} \spaced{\geq} \frac{S_1}{4} \cdot \min\inparen{1,\frac{S_1}{S_2}}
\]
\end{corollary}

\subsubsection*{Estimating $\abs{\Union A(\alpha, \beta)}$ via Inclusion-Exclusion}
\[
\abs{\Union_{\alpha, \beta}A(\alpha,\beta)}\spaced{\geq} \sum_{\alpha,\beta}\abs{A(\alpha, \beta)} \spaced{-} \sum_{(\alpha, \beta)\neq (\alpha',\beta')}\abs{A(\alpha, \beta) \intersection A(\alpha',\beta')}
\]

Let us first address the term $\sum \abs{A(\alpha, \beta)}$. 
As mentioned earlier, it is not an easy task to get a good handle on the set $A(\alpha, \beta)$ for polynomial such as $\NW$, for any reasonable monomial ordering. 
However, \cite{KS14} circumvent this difficult by using an indirect approach to estimate this term. 

For any derivative $\alpha$ and $\beta \in \partial_\alpha(P)$, define the set $S(\alpha, \beta)$ as the following set of multilinear monomials of degree $\ell$ that is disjoint from $\beta$. 
\[
S(\alpha, \beta) \spaced{=} \setdef{\gamma}{\begin{array}{c}\text{$\gamma$ is multilinear, has}\\\text{degree $\ell$ and $\gcd(\beta,\gamma)=1$ }\end{array}}
\]
This on the other hand is independent of any monomial ordering, and is also easy to calculate:
\[
\text{For every $\alpha, \beta$}\quad\quad \abs{S(\alpha, \beta)} \spaced{=} \binom{n - d + k}{\ell}.
\] 
\begin{lemma}[\cite{KS14}]\label{lem:As-to-Ss}
For any $\alpha$, 
\[
\sum_{\beta} \abs{A(\alpha, \beta)} \spaced{\geq} \abs{\Union_{\beta} S(\alpha, \beta)}
\]
\end{lemma}
\begin{proof}
Consider any $\gamma \in \Union_{\beta}S(\alpha, \beta)$. 
By definition, there is at least one non-multilinear monomial in $\gamma \cdot \partial_\alpha(P)$. 
Thus, in particular $\mathrm{LM}(\mathrm{mult}(\gamma \cdot \partial_\alpha(P))$ is non-zero and equal to some $\gamma \cdot \beta$ for some monomial $\beta \in \partial_\alpha(P)$. 
This also implies that $\gamma' = \gamma\cdot \beta \in A(\alpha, \beta)$. 
This yields an injective map $\phi$ 
\[
\phi:\Union_\beta S(\alpha,\beta) \spaced{\rightarrowtail} \setdef{(\beta, \gamma')}{\beta\in \partial_\alpha(P)\;,\;\gamma' \in A(\alpha, \beta)}
\] 
Since the size of the RHS is precisely $\sum_\beta \abs{A(\alpha, \beta)}$, the lemma follows. 
\end{proof}

\noindent
Thus, by another use of Inclusion-Exclusion on the $S(\alpha, \beta)$'s, we get
\begin{eqnarray*}
\abs{\Union_{\alpha, \beta}A(\alpha,\beta)}&\geq& \sum_{\alpha,\beta}\abs{A(\alpha, \beta)} \spaced{-} \sum_{(\alpha, \beta)\neq (\alpha',\beta')}\abs{A(\alpha, \beta) \intersection A(\alpha',\beta')}\\
 & \geq & \sum_\alpha \inparen{\sum_\beta \abs{S(\alpha, \beta)}} \spaced{-} \sum_\alpha \inparen{\sum_{\beta \neq \beta'}\abs{S(\alpha, \beta)\intersection S(\alpha,\beta')}}\\
 & & \quad\quad \spaced{-} \sum_{(\alpha, \beta)\neq (\alpha',\beta')}\abs{A(\alpha, \beta) \intersection A(\alpha',\beta')}
\end{eqnarray*}
Let us call the three terms in the RHS of the last equation as $T_1$, $T_2$ and $T_3$ respectively. 
Since we know the size of each $S(\alpha, \beta)$ exactly, the value of $T_1$ is easily obtained. 

\begin{lemma}[\cite{KS14}]\label{lem:T_1-value}
\begin{eqnarray*}
T_1(\alpha) \spaced{:=} \sum_{\beta}\abs{S(\alpha,\beta)}&=&\text{(\# mons in a deriv)} \cdot \binom{n-d+k}{\ell}
\end{eqnarray*}
\end{lemma}

\noindent
We shall be simplifying such binomial coefficients very often. 

\begin{lemma}\label{lem:binom-approx}
Let $n$ and $\ell$ be parameters such that $\ell = \frac{n}{2}(1 - \epsilon)$ for some $\epsilon = o(1)$. 
For any $a, b$ such that $a,b = O(\sqrt{n})$, 
\[
\binom{n - a}{\ell - b} \quad = \quad \binom{n}{\ell} \cdot 2^{-a} \cdot (1+\epsilon)^{a-2b} \cdot \exp(O(b\cdot \epsilon^2))
\]
\end{lemma}
\begin{proof}
The proof of the above lemma would repeated use the fact that $n! = (n-a)! \cdot n^a \cdot \poly(n)$ whenever $a = O(\sqrt{n})$ (see \cite[Lemma 3.4]{gkks13}). 
\begin{eqnarray*}
\frac{\binom{n-a}{\ell -b}}{\binom{n}{\ell}} & = & \frac{(n-a)!}{n!} \cdot \frac{\ell!}{(\ell -b)!}\cdot \frac{(n-\ell)!}{(n-\ell-a+b)!}\\
& \stackrel{\poly}{\approx}& \frac{1}{n^a} \cdot \ell^b \cdot \frac{(n-\ell)^a}{(n-\ell)^b}\\
& = & \frac{\inparen{\frac{n}{2}}^a(1 +\epsilon)^a}{n^a} \cdot \frac{(1-\epsilon)^{b}}{(1+\epsilon)^b}\\
& = & 2^{-a} \cdot (1+\epsilon)^{a - 2b} \cdot \exp(O(b\cdot \epsilon^2))
\end{eqnarray*}
\end{proof}

\noindent
Since our of parameters would be $\epsilon = \Theta\inparen{\frac{\log d}{\sqrt{d}}}$, the bound on $T_1$ can be simplified as
\begin{eqnarray*}
T_1(\alpha) & = & \text{(\# mons in a deriv)} \cdot \binom{n}{\ell} \cdot \inparen{\frac{1+\epsilon}{2}}^{d-k} \cdot \exp(-O(\log^2 d))\\
  & = & m^{e-k}\cdot \binom{n}{\ell} \cdot \inparen{\frac{1+\epsilon}{2}}^{d-k} \cdot \exp(-O(\log^2 d))\\
& =  &  \binom{n}{\ell} \cdot \exp(-O(\log^2 d))
\end{eqnarray*}
where we used the fact that every non-zero $k$-th order derivative of $\NW_{d,m,e}$ has exactly $m^{e-k}$ monomials and our setting of parameters. 

\begin{remark*}To avoid writing this factor of $\exp(O(\log^2 d))$, we shall use $\approx$ of $\gtrsim$ or $\lesssim$ to indicate that a factor $\exp(O(\log^2 d))$ is omitted. 
\end{remark*}

\bigskip

\noindent We now move on to the calculation of $T_2$. This is the first place where the choice of the polynomial and parameters becomes crucial. 

\begin{lemma}[\cite{KS14}]\label{lem:T_2-for-NW}
For the polynomial $P = \NW_{d,m,e}$, if $n = md$ and $\ell = \frac{n}{2}(1 - \epsilon)$ for $\epsilon = \Theta\inparen{\frac{\log d}{\sqrt{d}}}$, for any $\alpha \in \Delta$
\[
T_2(\alpha) \spaced{:=} \sum_{\beta\neq \beta'}\abs{S(\alpha, \beta)\intersection S(\alpha, \beta')} \quad \lesssim \quad m^{2(e-k)}\cdot \binom{n}{\ell} \cdot \inparen{\frac{1+\epsilon}{2}}^{2d -2k} 
\]
\end{lemma}
\begin{proof}
Recall that $S(\alpha,\beta) \intersection S(\alpha,\beta')$ is just set of all multilinear monomials $\gamma$ of degree $\ell$ that are disjoint from both $\beta$ and $\beta'$. Hence, for any pair of multilinear degree $(d-k)$ monomials $\beta \neq \beta' \in \partial_\alpha(P)$ such that $\deg(\gcd(\beta, \beta')) = t$, 
\[
\abs{S(\alpha, \beta)\intersection S(\alpha, \beta')} \spaced{=} \binom{n - 2d + 2k +t}{\ell}
\]
Thus, if we can count the number of pairs $(\beta, \beta')$ that agree on exactly $t$ places, we can obtain $T_2(\alpha)$. 
Note that for $\NW_{d,m,e}$, any two $\beta, \beta' \in\partial_\alpha(\NW_{d,m,e})$ can agree on at most $e-k$ places. 
Further, the number of pairs that agree in exactly $0\leq t\leq e-k$ places is at most
\[
m^{e-k} \cdot \binom{d-k}{t} \cdot (m-1)^{e-k-t}
\]
as there are $m^{e-k}$ choices for $\beta$, and $\binom{d-k}{t}$ choices for places where they may agree, and $(m-1)^{e-k-t}$ choices for $\beta'$ that agree with $\beta$ on those $t$ places. 
Thus,
\begin{eqnarray*}
T_2(\alpha) &\leq& \sum_{t=0}^{e-k} m^{e-k} \cdot \binom{d-k}{t} \cdot (m-1)^{e-k-t} \cdot  \binom{n - 2d + 2k +t}{\ell}\\
& \approx  & \sum_{t=0}^{e-k} m^{e-k} \cdot \binom{d-k}{t} \cdot (m-1)^{e-k-t} \cdot  \binom{n}{\ell} \frac{1}{2^{2d-2k -t}}\cdot (1+\epsilon)^{2d - 2k - t}\\
& \leq & m^{2(e-k)}\binom{n}{\ell}\inparen{\frac{1+\epsilon}{2}}^{2d -2k}\cdot\sum_{t=0}^{e-k}\binom{d-k}{t}\inparen{\frac{2}{(1+\epsilon)m}}^t\\
& \leq & m^{2(e-k)}\binom{n}{\ell}\inparen{\frac{1+\epsilon}{2}}^{2d -2k}\cdot \inparen{1+\frac{2}{(1+\epsilon)m}}^{d-k}\\
& = & m^{2(e-k)}\cdot \binom{n}{\ell} \cdot \inparen{\frac{1+\epsilon}{2}}^{2d -2k}\cdot O(1) \qquad\text{if $m = \Omega(d)$}\qedhere
\end{eqnarray*}
\end{proof}
\noindent

Combining this with \autoref{lem:T_1-value} and using Inclusion-Exclusion (\autoref{cor:inc-exc-str}), we get that for every $\alpha\in \Delta$,
\begin{eqnarray*}
\abs{\Union_{\beta} S(\alpha,\beta)} &\spaced{\gtrsim}& T_1(\alpha) \cdot \min\inparen{1,\frac{T_1(\alpha)}{T_2(\alpha)}}\\
& \approx & T_1(\alpha) \cdot \min\inparen{1,\frac{\pfrac{2}{1+\epsilon}^{d-k}}{m^{e-k}}}\\
& \approx & T_1(\alpha)
\end{eqnarray*}
by our choice of parameters. Note that $e$ needs to tailored very precisely to force the above condition! 
If $e$ is chosen too large or small, we get nothing from this whole exercise!

Thus by \autoref{lem:As-to-Ss} and \autoref{lem:T_1-value}, we get
\begin{equation}\label{eqn:T2-bound}
\sum_{\substack{\alpha\in \Delta\\\beta \in \partial_\alpha(P)}} \abs{A(\alpha,\beta)} \spaced{\geq} \abs{\Delta} \cdot \abs{\Union_{\beta} S(\alpha,\beta)} \spaced{\geq} \abs{\Delta} \cdot T_1(\alpha) \spaced{\approx} \abs{\Delta} \cdot \binom{n}{\ell}
\end{equation}

\subsubsection*{Upper bounding $\sum \abs{A(\alpha,\beta)\intersection A(\alpha',\beta')}$}

We are still left with the task of upper bounding
\[
T_3 \quad = \quad \sum_{(\alpha, \beta)\neq (\alpha',\beta')} \abs{A(\alpha, \beta) \intersection A(\alpha',\beta')}
\]
As mentioned earlier, we really do not have a good handle on the set $A(\alpha, \beta)$, and certainly not on the intersection of two such sets. 
Once again, we shall use a proxy that is easier to estimate to upper bound $T_3$. 

The set $A(\alpha, \beta) \intersection A(\alpha',\beta')$ consists of multilinear monomials $\gamma$ of degree $\ell + d -k$ such that there exists multilinear monomials $\gamma', \gamma''$ of degree $\ell$ satisfying
\begin{eqnarray*}
\gamma & = & \gamma' \beta \spaced{=} \gamma'' \beta',\\
 \gamma'\beta & = & \mathrm{LM}(\mathrm{mult}(\gamma' \partial_\alpha(P)))\\
\text{and}\quad \gamma''\beta' & = & \mathrm{LM}(\mathrm{mult}(\gamma'' \partial_{\alpha'}(P)))
\end{eqnarray*}
This in particular implies that $\gamma$ must be divisible by both $\beta$ and $\beta'$. 

\begin{observation}\label{obs:T3-proxy}
If $\deg(\gcd(\beta, \beta')) = t$, then
\[
\abs{A(\alpha, \beta) \intersection A(\alpha', \beta')} \spaced{\leq} \binom{n - 2d + 2k + t}{\ell - d + k +t}
\]
\end{observation}
\begin{proof}
Every monomial $\gamma \in A(\alpha, \beta) \intersection A(\alpha', \beta')$ must be divisible by $\beta$ and $\beta'$. 
Since $\abs{\beta \union \beta'} = 2d - 2k - t$, the number of choices of $\gamma$ is precisely
\[
\binom{n - (2d - 2k -t)}{(\ell + d - k) - (2d - 2k - t)} \quad = \quad \binom{n - 2d + 2k + t}{\ell - d + k + t}\qedhere
\]
\end{proof}

One needs a similar argument as in the case of $T_2$ to figure out how many pairs $(\alpha, \beta) \neq (\alpha',\beta')$ are there with $\deg(\gcd(\beta, \beta')) = t$ and sum them up accordingly. 

\begin{lemma}[\cite{KS14}] \label{lem:T3-bound}
For the polynomial $\NW_{d,m,e}$, and $n = md$ and $\ell = \frac{n}{2}(1 - \epsilon)$ for $\epsilon = \Theta\inparen{\frac{\log d}{\sqrt{d}}}$, 
\[
\sum_{(\alpha,\beta)\neq (\alpha',\beta')} \abs{A(\alpha,\beta) \intersection A(\alpha',\beta')} \quad \lesssim \quad \abs{\Delta}^2 \cdot \pfrac{m^{e-k}}{2^{d - k}}^2 \cdot \binom{n}{\ell}\cdot
\]
\end{lemma}
\begin{proof}
Fix a pair of derivatives $\alpha,\alpha'$. Let
\[
T_3(\alpha,\alpha') \spaced{:=} \sum_{\substack{\beta \in \partial_\alpha(P)\\
\beta'\in \partial_{\alpha'}(P)\\
(\alpha,\beta)\neq (\alpha',\beta')}} \abs{A(\alpha,\beta)\intersection A(\alpha',\beta')}
\]
As before, we shall first count the number of pairs of monomials $\beta \in \partial_\alpha P$ and $\beta' \in \partial_{\alpha'} P$ such that $\gcd(\beta, \beta') = t$. 
Note that since $\alpha$ may differ from $\alpha'$, we could potentially have $\gcd(\beta_1,\beta_2) = e$. 
Once again, this is easily seen to be at most
\[
m^{e-k} \cdot \binom{d-k}{t} \cdot (m-1)^{e-k-t}. 
\]
\noindent Therefore, using \autoref{obs:T3-proxy}, 
\begin{eqnarray*}
T_3(\alpha, \alpha') & \leq & \sum_{t=0}^{e} m^{e-k} \cdot (m-1)^{e-k -t} \binom{d-k}{t} \binom{n- 2d + 2k +t}{\ell - d + k + t}\\
& \approx & \sum_{t=0}^{e} m^{e-k} \cdot (m-1)^{e-k -t} \binom{d-k}{t} \cdot \binom{n}{\ell} \inparen{\frac{1}{2}}^{2d - 2k -t}  (1+\epsilon)^{t}\\
& \leq & \frac{m^{2(e-k)}}{2^{2(d-k)}} \cdot \binom{n}{\ell} \cdot \inparen{1 + \frac{2(1+\epsilon)}{m}}^{d-k}\\
& \approx & \pfrac{m^{e-k}}{2^{d-k}}^2 \cdot \binom{n}{\ell} \\
\implies T_3 & \lesssim & \abs{\Delta}^2 \cdot \pfrac{m^{e-k}}{2^{d-k}}^2 \cdot \binom{n}{\ell}
\end{eqnarray*}
\end{proof}
\noindent
Recalling that we have chosen our parameters so that 
\[
\frac{m^{e-k}}{2^{d-k}} \spaced{\approx} \inparen{\frac{1}{1+\epsilon}}^{d-k} \quad\text{and}\quad \abs{\Delta} = (1+\epsilon)^{2(d-k)},
\]
the above equation reduces to 
\[
T_3 \spaced{=}\sum_{(\alpha,\beta)\neq (\alpha',\beta')}\abs{A(\alpha,\beta)\intersection A(\alpha',\beta')} \quad \lesssim \quad \abs{\Delta} \cdot  \binom{n}{\ell}.
\]
Combining with \eqref{eqn:T2-bound}, we obtain the required bound for $\abs{\Union A(\alpha, \beta)}$ via Inclusion-Exclusion (\autoref{cor:inc-exc-str}). 
\[
\GammaPSD_{k,\ell}(\NW_{d,m,e}) \spaced{\geq} \abs{\Union_{\alpha,\beta} A(\alpha,\beta)} \spaced{\gtrsim} \binom{n}{\ell} \cdot (1+\epsilon)^{2d - 2k}
\]
The only thing left to observe is that by \autoref{lem:binom-approx},
\[
\binom{n}{\ell + d - k} \spaced{\approx} \binom{n}{\ell} \cdot (1+\epsilon)^{2d - 2k}
\]
and that completes the proof of \autoref{lem:KS-tight-bound}. \hfill\qed

\end{document}